%% file: paper-alt.tex
\documentclass[sigconf,10pt,dvipsnames]{acmart}

\setcopyright{none}
\renewcommand\footnotetextcopyrightpermission[1]{}
\usepackage{alltt}
\usepackage{balance}
\usepackage{enumitem}
\setlist{noitemsep}
\usepackage{graphicx}
\usepackage{listings}
\usepackage{tikz}
\usepackage{pgfplots}
\usepackage{bm}
\usepackage[makeroom]{cancel}

\usetikzlibrary{arrows,calc}

\newcommand*{\tv}[2]{\big[{#1\atop #2}\big]}
\newcommand*{\hlchange}[1]{{\color{RawSienna}\bm{#1}}}
\newcommand*{\hlcrossoff}[1]{{\color{Tan}\bm{\xcancel{#1}}}}
\newcommand*{\smallsf}[1]{\textsf{\small #1}}
\newcommand*{\onef}[0]{\mathbf{1}}
\newcommand*{\minarity}[0]{\mu}

\lstdefinelanguage{MyPseudoCode}{
  columns=flexible,
  keepspaces=true,
  basicstyle=\footnotesize\ttfamily,
  morekeywords={and,each,else,false,for,fun,if,in,not,null,or,return,true,while},
  keywordstyle=\bfseries,
  identifierstyle=\slshape,
  numbers=left,
  numberstyle=\sffamily\tiny,
  numbersep=2mm,
  xleftmargin=4mm,
  mathescape=true,
  escapeinside={(*@}{@*)},
}
\newcommand*{\myparagraph}[2][\medskip]{#1
  \noindent\textbf{#2}.}

\newcommand{\TITLE}{Sub-O(log n) Out-of-Order\\ Sliding-Window Aggregation}

\newtheorem{theorem}{Theorem}[section] 

\newtheorem{lemma}[theorem]{Lemma}

\newcommand*{\dsName}[0]{FiBA}

\begin{document}

\title[Sub-O(log n) Out-of-Order Sliding-Window Aggregation]{\TITLE}

\author{Kanat Tangwongsan}
\affiliation{\institution{{\small Mahidol University International College}}}
\email{kanat.tan@mahidol.edu}

\author{Martin Hirzel}
\affiliation{\institution{IBM Research}}
\email{hirzel@us.ibm.com}

\author{Scott Schneider}
\affiliation{\institution{IBM Research}}
\email{scott.a.s@us.ibm.com}

\input{abstract}

\maketitle

\input{intro}
\input{problem}
\input{btree}
\input{sharing}
\input{results}
\input{related}
\input{concl}

\bibliographystyle{ACM-Reference-Format}
\bibliography{ref}
\appendix
\input{appdx-lb}
\input{appdxproof}
\balance
\input{amortized_cost}

\end{document}

%% file: abstract.tex
\begin{abstract}
  Sliding-window aggregation summarizes the most recent information in
  a data stream. Users specify how that summary is computed, usually
  as an associative binary operator because this is the most general
  known form for which it is possible to avoid na\"ively scanning
  every window. For strictly in-order arrivals, there are algorithms
  with $O(1)$ time per window change assuming associative operators.
  Meanwhile, it is common in practice for streams to have data
  arriving slightly out of order, for instance, due to clock drifts or
  communication delays. Unfortunately, for out-of-order streams, one
  has to resort to latency-prone buffering or pay $O(\log n)$ time per
  insert or evict, where $n$ is the window size.

  This paper presents the design, analysis, and implementation of
  \dsName{}, a novel sliding-window aggregation algorithm with an
  amortized upper bound of $O(\log d)$ time per insert or
  evict, where $d$ is the distance of the inserted or evicted value to the
  closer end of the window.
  This means $O(1)$ time for in-order arrivals and nearly $O(1)$ time
  for slightly out-of-order arrivals, with a smooth transition towards
  $O(\log n)$ as $d$ approaches~$n$. We also prove a matching lower
  bound on running time, showing optimality.
  Our algorithm is as general as the prior state-of-the-art: it
  requires associativity, but not invertibility nor commutativity.
  At the heart of the algorithm is a careful combination of
  finger-searching techniques, lazy rebalancing, and
  position-aware partial aggregates. We further show how to answer
  range queries that aggregate subwindows for window sharing. Finally,
  our experimental evaluation shows that \dsName{} performs well
  in practice and supports the theoretical findings.
  %
  %
  %
\end{abstract}


%% file: intro.tex
\section{Introduction}
\label{sec:intro}

Stream processing is now in widespread production use in domains as varied as
telecommunication, personalized advertisement, medicine, transportation, and
finance. It is generally the paradigm of choice for applications that expect
high throughput and low latency. Regardless of domain, nearly every stream
processing application involves some form of aggregation or another, with one of
the most common being sliding-window aggregation.

Sliding-window aggregation derives a summary statistic over a user-specified
amount of recent streaming data. Users also define how that summary statistic is
computed, usually in the form of an associative binary operator~\cite{boykin_et_al_2014}, as that is the
most general known form for which computation can be effectively incrementalized
to avoid na\"ively scanning every window. While some associative
aggregation operators, such as sum, are also invertible, many, such as maximum
or Bloom filters, are merely associative but not invertible.

Recent algorithmic research on sliding-window aggregation has given much
attention to streams with strictly in-order arrivals. The standard interface for
sliding-window aggregation supports insert, evict, and query. In the in-order
setting, there are
algorithms~\cite{shein_chrysanthis_labrinidis_2017,tangwongsan_hirzel_schneider_2017}
for associative operators that take only $O(1)$ time per window change, without
requiring the operator to be invertible nor commutative.

In reality, however, out-of-order streams are the norm~\cite{akidau_et_al_2013}.
Clock drift and disparate latency in computation and communication, for example,
can cause values in a stream to arrive in a different order than their
timestamps. Processing out-of-order streams is already supported in many stream
processing platforms
(e.g.,~\cite{akidau_et_al_2013,zaharia_et_al_2013,carbone_et_al_2015,akidau_et_al_2015}).
Still, in terms of performance, users who want the full generality of
associative operators have to resort to latency-prone buffering or,
alternatively, use an augmented balanced tree, such as a B-tree, at a cost of
$O(\log n)$ time per insert or evict, where $n$ is the window size. This stands
in stark contrast with the in-order setting, especially for when the streams are
nearly in order. Thus, we ask whether there exists a sub-$O(\log n)$ algorithm
for out-of-order streams; this paper is our affirmative answer.

This paper introduces the finger B-tree aggregator (\dsName{}), a novel
algorithm that efficiently aggregates sliding windows on out-of-order streams
and in-order streams alike. Each insert or evict takes amortized $O(\log d)$
time\footnote{See Theorem~\ref{trm_finger_complexity} for a more formal
  statement.}, where the \emph{out-of-order distance} $d$ is the distance from the
inserted or evicted value to the closer end of the window. The $O(\log d)$
complexity means $O(1)$ for in-order streams, nearly $O(1)$ for slightly
out-of-order streams, and never more than $O(\log n)$ even for severely
out-of-order streams. The worst-case time for any one particular insert or evict
is $O(\log n)$, which only happens in the rare case of rebalancing all the way
up the tree. \dsName{} requires $O(n)$ space and takes $O(1)$ time for a
whole-window query. Furthermore, it is as general as the prior state-of-the-art,
supporting variable-sized windows and only requiring associativity from the
operator.


Our solution can be summarized as finger B-trees~\cite{guibas_et_al_1977} with position-aware partial
aggregates. Starting with the classic B-trees, we first add pointers, or
\emph{fingers}, to the start and end of the tree. These fingers make it possible
to perform the search for the value to insert or evict in $O(\log d)$ worst-case
time. Second, we adapt a specific variant of B-trees where the rebalance to fix
the size invariants takes amortized $O(1)$ time; specifically, we use B-trees
with \lstinline{MAX_ARITY$=2\cdot$MIN_ARITY} and where rebalancing happens
after-the-fact~\cite{huddleston_mehlhorn_1982}. Third and most importantly, we
develop novel position-aware partial aggregates and a corresponding algorithm to
bound the cost of aggregate repairs to the cost of search plus rebalance.

The running time of \dsName{} is asymptotically the best possible in
general. We prove a lower bound showing that for insert and evict
operations with out-of-order distance up to $d$, the amortized cost of
an operation in the worst case must be at least $\Omega(\log d)$.

Furthermore, we show how \dsName{} can support window sharing with query time
logarithmic in the subwindow size and the distance from the largest window's
boundaries. Here, the space complexity is $O(n_{\max})$, where $O(n_{\max})$ is
the size of the largest window.

Our experiments confirm the theoretical findings and show that \dsName{}
performs well in practice. For out-of-order streams, it is a substantial
improvement over existing algorithms in terms of both latency and throughput.
For strictly in-order streams (i.e., FIFO), it demonstrates constant time
performance and remains competitive with specialized algorithms for in-order
streams.


We hope \dsName{} will be used to make streaming applications less
resource-hungry and more responsive for out-of-order streams.


%% file: problem.tex
\section{Problem Statement: OoO SWAG}
\label{sec:problem}

This section states the problem addressed in this paper more formally.
Consider a data stream where each value carries a logical time in the form of a
timestamp. Throughout, we denote a timestamped value as $\tv{t}{v}$. For
example, $\tv{17}{4}$ is the value $4$ at logical time~$17$. The
examples in this paper use natural numbers for timestamps, but our
algorithms do not depend on any properties of the natural numbers
besides being totally ordered. For instance, our algorithms work
just as well with date/time representations or with real numbers.

It is intuitive to assume that values in such a stream arrive in nondecreasing
order of time (in order). However, due to clock drift and disparate latency in
computation and communication, among other factors, values in a stream often
arrive in a different order than their timestamps. Such a stream is said to have
\emph{out-of-order} (OoO) arrivals---there exists a later-arriving value that
has an earlier logical time than a previously-arrived value.

Our goal in this paper is to maintain the aggregate value of a time-ordered
sliding window in the face of out-of-order arrivals. To motivate our formulation
below, consider the following example, which maintains the \textsf{max}
and the \textsf{maxcount}, i.e., the number of times the \textsf{max} occurs in
the sliding window.

\vspace*{1mm}
$\tv{17}{4},\tv{19}{3},\tv{20}{0},\tv{21}{4}$\hfill$\textsf{max}\;4,\textsf{maxcount}\;2$
\vspace*{1mm}

Initially, the values \mbox{$4,3,0,4$} arrive in the same order as their
associated timestamps \mbox{$17,19,20,21$}. The maximum value is $4$, and
\textsf{maxcount} is $2$ because $4$ occurs twice. When stream values arrive in
order, they are simply appended. For instance, when $\tv{22}{4}$ arrives, it is
inserted at the end:

\vspace*{1mm}
$\tv{17}{4},\tv{19}{3},\tv{20}{0},\tv{21}{4},\hlchange{\tv{22}{4}}$\hfill$\textsf{max}\;4,\textsf{maxcount}\;3$
\vspace*{1mm}

However, when values arrive out-of-order, they must be inserted into the
appropriate spots to keep the sliding window time-ordered. For instance, when
$\tv{18}{5}$ arrives, it is inserted between timestamps $17$ and $19$:

\vspace*{1mm}
$\tv{17}{4},\hlchange{\tv{18}{5}},\tv{19}{3},\tv{20}{0},\tv{21}{4},\tv{22}{4}$\hfill$\textsf{max}\;5,\textsf{maxcount}\;1$
\vspace*{1mm}

As for eviction, stream values are usually removed from a window in order, for
instance, evicting $\tv{17}{4}$ from the front:

\vspace*{1mm}
$\hlcrossoff{\tv{17}{4}},\tv{18}{5},\tv{19}{3},\tv{20}{0},\tv{21}{4},\tv{22}{4}$\hfill$\textsf{max}\;5,\textsf{maxcount}\;1$
\vspace*{1mm}

Notice that, in general, eviction cannot always be accomplished by simply
inverting the aggregation value. For instance, evicting $\tv{18}{5}$ cannot be
done by ``subtracting off'' the value $5$ from the current aggregation value.
The algorithm needs to efficiently discover the new \textsf{max}~4 and
\textsf{maxcount}~2:

\vspace*{1mm}
$\hlcrossoff{\tv{18}{5}},\tv{19}{3},\tv{20}{0},\tv{21}{4},\tv{22}{4}$\hfill$\textsf{max}\;4,\textsf{maxcount}\;2$
\vspace*{1mm}

\medskip

\noindent{}\textbf{Monoids.} There are other streaming aggregations
besides \textsf{max} and \textsf{maxcount}. Monoids capture a large class of
commonly used aggregations~\cite{boykin_et_al_2014,tangwongsan_et_al_2015}. A
\emph{monoid} is a triple $\mathcal{M} = (S, \otimes, \onef{})$, where $\otimes:
S \times S \to S$ is a binary associative operator on $S$, with $\onef{}$ being
its identity element. Notice that $\otimes$ only needs to be associative; it
does not need not be commutative or invertible. For example, to express
\textsf{max} and \textsf{maxcount} as a monoid, if $m$ and $c$ are the max and
maxcount, then

\begin{equation*}
  \textstyle
  \langle m_1,c_1\rangle\otimes_{\mathsf{max},\mathsf{maxcount}}\langle m_2,c_2\rangle
  = \left\{\begin{array}{l@{\;}l}
             \langle m_1,c_1\rangle     & \mathrm{if}\; m_1> m_2\\
             \langle m_2,c_2\rangle     & \mathrm{if}\; m_1< m_2\\
             \langle m_1,c_1+c_2\rangle & \mathrm{if}\; m_1= m_2
           \end{array}\right.
\end{equation*}

Since $\otimes$ is associative, no parentheses are needed for repeated
application. When the context is clear, we even omit $\otimes$, for example,
writing \textsf{qstu} for
\mbox{$\textsf{q}\otimes\textsf{s}\otimes\textsf{t}\otimes\textsf{u}$}. This
concise notation is borrowed from the mathematicians' convention of omitting
explicit multiplication operators.

\medskip
\noindent\textbf{OoO SWAG.} This paper is concerned with maintaining an aggregation on a
time-ordered sliding window where the aggregation operator can be expressed as a
monoid. This can be formulated as an abstract data type (ADT) as follows:
\begin{definition}
  Let $(\otimes, \onef{})$ be a binary operator operator from a monoid and its
  identity. The \emph{out-of-order sliding-window aggregation} (OoO SWAG) ADT is
  to maintain a time-ordered sliding window $\tv{t_1}{v_1}, \dots,
  \tv{t_n}{v_n}$, $t_i < t_{i+1}$, supporting the following operations:
\begin{itemize}[label=---, topsep=2pt,leftmargin=1.5\parindent]
\item \lstinline{insert($t$ : Time, $v$ : Agg)} checks whether $t$ is already in
  the window, i.e., whether there is an $i$ such that \mbox{$t=t_i$}. If so, it
  replaces $\tv{t_i}{v_i}$ by \mbox{$\tv{t_i}{v_i\otimes v}$}. Otherwise, it
  inserts \mbox{$\tv{t}{v}$} into the window at the appropriate location.
\item \lstinline{evict($t$ : Time)} checks whether $t$ is in the window, i.e.,
  whether there is an $i$ such that \mbox{$t=t_i$}. If so, it removes
  \mbox{$\tv{t_i}{v_i}$} from the window. Otherwise, it does nothing.
\item \lstinline{query() : Agg} combines the values in time order using the
  $\otimes$ operator. In other words, it returns \mbox{$v_1\otimes\ldots\otimes
    v_n$} if the window is non-empty, or $\onef{}$ if empty.
\end{itemize}
\end{definition}

\noindent\textbf{Lower Bound.}
How fast can OoO SWAG operations be supported? For in-order streams, the SWAG
operations can be handled in $O(1)$ time per
operation~\cite{tangwongsan_hirzel_schneider_2017,shein_chrysanthis_labrinidis_2017}.
But the problem becomes more difficult when the stream has out-of-order
arrivals. We prove in this paper that to handle out-of-order distance up to $d$,
the amortized cost of a OoO SWAG operation in the worst case must be at least
$\Omega(\log d)$.

\begin{theorem}
  \label{thm:swag-lb}
  Let $m, d \in \mathbb{Z}$ be given such that $m \geq 1$ and $0 \leq d \leq m$.
  For any OoO SWAG algorithm, there exists a sequence of $3m$ operations, each
  with out-of-order distance at most $d$, for which the algorithm requires a total
  of at least $\Omega(m \log (1+d))$ time.
\end{theorem}
The proof, which appears in Appendix~\ref{sec:appdx-lb}, shows this in two
steps. First, it establishes a sorting lower bound for permutations on
$m$ elements with out-of-order distance at most~$d$. Second, it gives a
reduction proving that maintaining
OoO SWAG is no easier than sorting such permutations.

\medskip

\noindent\textbf{Orthogonal Techniques.}
OoO SWAG operations are designed to work well with other stream aggregation
techniques.

The \lstinline{insert($t,v$)} operation supports the case where $t$ is already
in the window, so it works with pre-aggregation schemes such as window
panes~\cite{li_et_al_2005}, paired
windows~\cite{krishnamurthy_wu_franklin_2006}, cutty
windows~\cite{carbone_et_al_2016}, or Scotty~\cite{traub_et_al_2018}. For
instance, for a 5-hour sliding window that advances in 1-minute increments, the
logical times can be rounded to minutes, leading to more cases where $t$ is
already in the window.

The \lstinline{evict($t$)} operation accommodates the case where $t$
is not the oldest time in the window, so it works with streaming
systems that use retractions
\cite{abadi_et_al_2005,akidau_et_al_2013,akidau_et_al_2015,barga_et_al_2007,brito_et_al_2008,chandramouli_goldstein_maier_2010,li_et_al_2008,zaharia_et_al_2013}.

Neither \lstinline{insert($t,v$)} nor \lstinline{evict($t$)} are limited to
values of $t$ that are near either end of the window, so they work in the
general case, not just in cases where the out-of-order distance is bounded by
buffer sizes or low watermarks.

\medskip

\noindent{}\textbf{Query Sharing.}
As defined above, OoO SWAG does not support query sharing. However, query sharing for
different window sizes can be accommodated via a range query:
\begin{itemize}[label=---, topsep=2pt,leftmargin=1.5\parindent]

\item \lstinline{query($t_\texttt{from}$ : Time, $t_\texttt{to}$ : Time) : Agg }
  aggregates exactly the values from the window whose times fall between
  $t_\texttt{from}$ and~$t_\texttt{to}$. That is, it returns
  \mbox{$v_{i_\texttt{from}}\otimes\ldots\otimes v_{i_\texttt{to}}$}, where
  $i_\texttt{from}$ is the largest $i$ such that $t_\texttt{from}\le
  t_{i_\texttt{from}}$ and $i_\texttt{to}$ is the smallest $i$ such that
  $t_{i_\texttt{to}}\le t_\texttt{to}$. If the subrange contains no values, the
  operation returns~$\onef{}$.
\end{itemize}


In these terms, the problem statement of this paper is
\begin{quote}
  to design and implement efficient OoO SWAG operations as well as range-query support
  for arbitrary monoids~$(S,\otimes,\onef{})$.
\end{quote}



%% file: btree.tex
\section{Finger B-Tree Aggregator (FiBA)}
\label{sec:btree}

This section introduces our algorithm gradually, giving intuition along the way.
It begins by describing a basic algorithm (Section~\ref{sec:classic}) that
utilizes a B-tree augmented with aggregates. This algorithm takes $O(\log n)$
time for each \lstinline{insert} or \lstinline{evict} operation. Reducing the
time complexity below $\log n$ requires further observations and ideas. This is
explored intuitively in Section~\ref{sec:intuition} with details fleshed out in
Section~\ref{sec:algorithm}.

\subsection{Basic Algorithm: Augmented B-Tree}
\label{sec:classic}

One way to implement the OoO SWAG is to start with a classic \mbox{B-tree} with
timestamps as keys and augment that tree with aggregates. This is a baseline
implementation, which will be built upon. Even though any balanced trees can, in
fact, be used, we chose the B-tree because it is well-studied and has
customizable fan-out degree, providing opportunities for experimentation.

There are many B-tree variations. The range of permissible
\emph{arity}, or fan-out degree of a node, is controlled by two
parameters \lstinline{MIN_ARITY} and \lstinline{MAX_ARITY}.
While \lstinline{MIN_ARITY} can be any integer
greater or equal to $2$, most B-tree variations require that
\lstinline{MAX_ARITY} be at least $2\cdot\mbox{\lstinline{MIN_ARITY}} - 1$.
Hence, if $a(y)$---or simply $a$ when the context is clear---denotes the arity
of a node $y$, then a B-tree obeys the following \emph{size invariants}:
\begin{itemize}[leftmargin=1em]
\item For a non-root node $y$, \lstinline{MIN_ARITY$\,\le a(y)$}; for the root,
  $2\le a$.
  \item For all nodes, \lstinline{$a\le\,$MAX_ARITY}.
  \item All nodes have $a-1$ timestamps and values
    $\tv{t_0}{v_0},\ldots,\tv{t_{a-2}}{v_{a-2}}$.
  \item All non-leaf nodes have $a$ child pointers $z_0,\ldots,z_{a-1}$.
\end{itemize}

\begin{figure}[tb]
  \centerline{\includegraphics[scale=0.41]{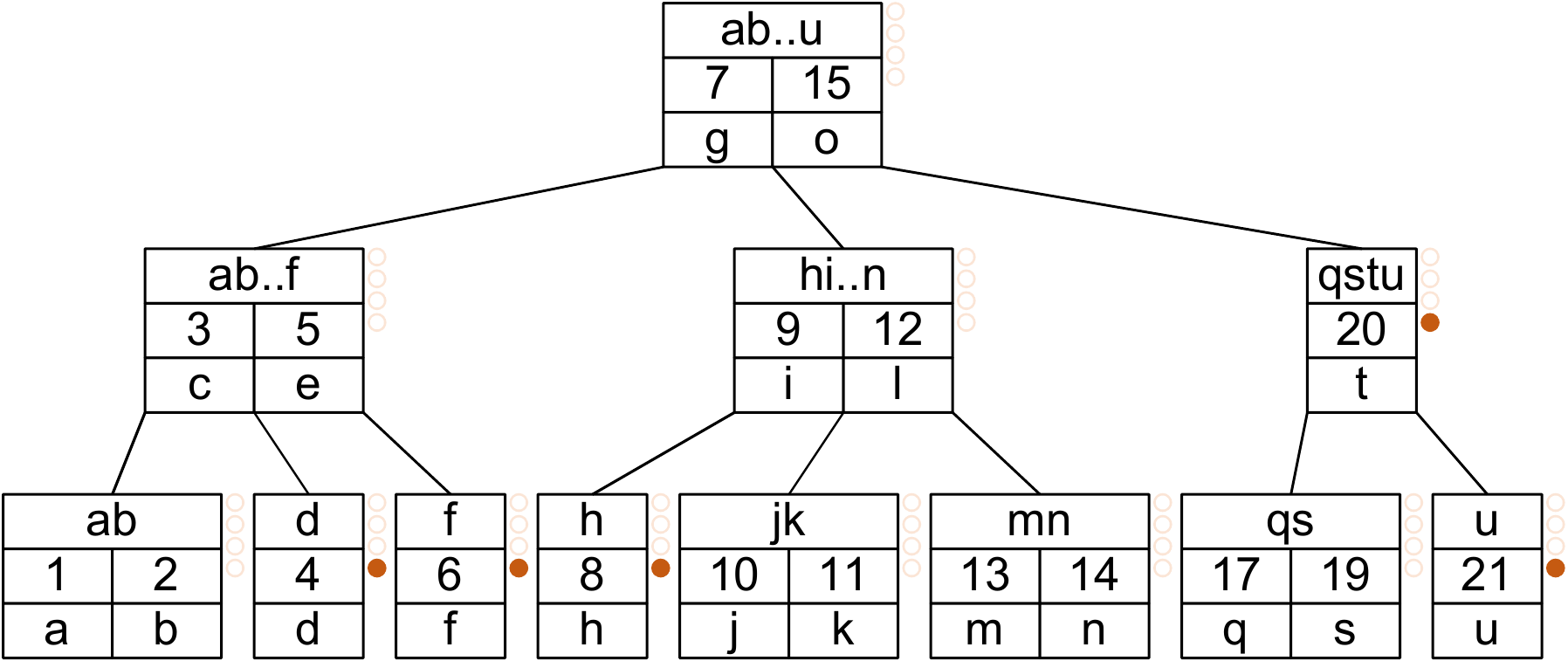}}
  \vspace*{-3mm}
  \caption{\label{fig_classic}Classic B-tree augmented with aggregates.}
\end{figure}

Figure~\ref{fig_classic} illustrates a B-tree augmented with aggregates. In this
example, \lstinline{MIN_ARITY} is~2 and \lstinline{MAX_ARITY} is
$2\cdot\mbox{\lstinline{MIN_ARITY}} = 4$. Consequently, all nodes have 1--3
timestamps and values, and non-leaf nodes have 2--4 children. Each node in the
tree contains an aggregate, an array of timestamps and values, and optionally
pointers to the children. For instance, the root node contains the aggregate
\smallsf{ab..u}, the values and their timestamps
\mbox{$\mathsf{\tv{7}{g},\tv{15}{o}}$}, and pointers to three children. Because
we use timestamps as keys, the entries are time-ordered, both within a node and
across nodes, with timestamps stored in a parent node separating and limiting the time in
the subtrees it points to. The tree is always height-balanced. Additionally, all
leaves are at the same depth.

\emph{What aggregate is kept in a node?} For each node~$y$, the aggregate
$\Pi_\uparrow(y)$ stored at that node obeys the up-aggre\-gation invariant:
\[\Pi_\uparrow(y)=\Pi_\uparrow(z_0)\otimes v_0\otimes\Pi_\uparrow(z_1)\otimes\ldots\otimes v_{a-2}\otimes\Pi_\uparrow(z_{a-1})\]
By a standard inductive argument, $\Pi_\uparrow(y)$ is the aggregation of the
values inside the subtree rooted at $y$. This means the \lstinline{query()}
operation can simply return the aggregation value at the root
(\lstinline{root.agg}).

The operations \lstinline{insert($t,v$)} or \lstinline{evict($t$)} first search
for the node where $t$ belongs. Second, they locally insert or evict at that
node, updating the aggregate stored at that node. Then, they rebalance the tree
starting at that node and going up towards the root as necessary to fix any size
invariant violations, while also repairing aggregate values along the way.
Finally, they repair any remaining aggregate values not repaired during
rebalancing, starting above the node where rebalancing topped out and visiting
all ancestors up to the root.

\begin{theorem}\label{trm_classic_correctness}
  In a classic B-tree augmented with aggregates, if it stores
  \mbox{$\tv{t_1}{v_1},\ldots,\tv{t_n}{v_n}$}, the operation \lstinline{query()}
  returns \mbox{$v_1\otimes\ldots\otimes v_n$}.
\end{theorem}

\begin{proof}
  After each operation, all nodes obey the aggregation invariant, and
  \lstinline{$\Pi_\uparrow($root$)$} contains
  \mbox{$v_1\otimes\ldots\otimes v_n$}.
\end{proof}

\begin{theorem}\label{trm_classic_complexity}
  In a classic B-tree augmented with aggregates, the operation
  \lstinline{query()} costs at most~$O(1)$ time and operations
  \lstinline{insert($t,v$)} or
  \lstinline{evict($t$)} take at most $O(\log n)$ time.
\end{theorem}

\begin{proof}
  As is standard, we treat the arity of a node as bounded by a constant. The
  query operation and the local insert or evict visit only a single node. The
  search, rebalance, and repair visit at most two nodes per tree level. The work
  is thus bounded by the tree height, which is $O(\log n)$ since the tree is
  height-balanced~\cite{bayer_mccreight_1972,cormen_leiserson_rivest_1990,huddleston_mehlhorn_1982}.
  Hence, the total cost per operation is $O(\log n)$.
\end{proof}

\subsection{Breaking the $O(\log n)$ Barrier}
\label{sec:intuition}

The basic algorithm just described supports OoO SWAG operations in $O(\log n)$
time using an augmented classic B-tree. To improve upon the time complexity, we
now discuss the bottlenecks in the basic algorithm and outline a plan to resolve
them.

In the basic algorithm, the \lstinline{insert($t,v$)} and \lstinline{evict($t$)}
operations involve four steps: (1)~search for the node where $t$ belongs;
(2)~locally insert or evict; (3)~rebalance to repair size invariants; and
(4)~repair
remaining aggregation invariants. If one treats arity as constant, the local
insertion or eviction operation takes constant time, as does the
\lstinline{query()} operation. But each of the steps for search, rebalance, and
repair takes up to $O(\log n)$ time. Hence, these are the bottleneck steps and
will be improved upon as follows:
\begin{enumerate}[label=(\roman*)]
\item By maintaining ``fingers'' to the leftmost and rightmost leaves, we will
  reduce the search complexity to $O(\log d)$, where $d$ is the distance to the
  closer end of the sliding-window boundary. This means that in the FIFO or
  near-FIFO case, the search complexity will be constant.
\item By using an appropriate \lstinline{MAX_ARITY} and a somewhat lazy strategy
  for rebalancing, we will make sure that rebalance takes no more than constant
  in the amortized sense. This means that for any operation that affects the
  tree structure, the cost to restore the proper tree structure amounts to
  constant per operation, regardless of out-of-order distance.
\item By introducing position-dependent aggregates, we will ensure that repairs
  to the aggregate values are made only to nodes along the search path or
  involved in restructuring. This means that the repairs cost no more than the
  cost of search and rebalance.
\end{enumerate}

We combine the above ideas into a novel sub-$O(\log n)$ algorithm for OoO SWAG.
Below, we describe how these ideas will be implemented intuitively, leaving
detailed algorithms and proofs to Section~\ref{sec:algorithm}.

\medskip

\noindent{}\textbf{Sub-$O(\log n)$ Search.}
In classic B-trees, a search starts at the root and ends at the node being
searched, henceforth called~$y$. Often, $y$ is a leaf, so the search visits
$O(\log n)$ nodes. However, instead of starting at the root, one can start at
the left-most or right-most leaf in the tree. This requires pointers to the
left-most or right-most leaf, henceforth called the left and right
fingers~\cite{guibas_et_al_1977}. In addition, we keep a parent pointer at each
node. Hence, the search can start at the nearest finger, walk up to the nearest
common ancestor of the finger and~$y$, and walk down from there to~$y$. The
resulting algorithm runs in $O(\log d)$, where~$d$ is the distance from the
nearest end of the window--or more precisely, $d$ is the number of timed values
from~$y$ to the nearest end of the window.

\smallskip

\noindent\textbf{Sub-$O(\log n)$ Rebalance.}
Insertions and evictions can cause nodes to overflow or underflow, thus
violating the size invariants. There are two popular strategies that address
this: either before or after the fact. The before-the-fact strategy ensures that
ancestors of the affected node are not at risk of overflow or underflow by
preventive rebalancing, so that the arity~$a$ is at least one further away from
the threshold required by the size invariants
(e.g.,~\cite{cormen_leiserson_rivest_1990}). The after-the-fact strategy first
performs the local insert or evict step, then repairs any resulting overflow or
underflow to ensure the size invariants hold again by the end of the entire
insert or evict operation. We adopt the after-the-fact strategy, which has been
shown to take amortized constant time~\cite{huddleston_mehlhorn_1982} as long as
$\mbox{\lstinline{MAX_ARITY}} \geq 2\cdot{}\mbox{\lstinline{MIN_ARITY}}$. For
simplicity, we use $\mbox{\lstinline{MAX_ARITY}} =
2\cdot{}\mbox{\lstinline{MIN_ARITY}}$. The amortized cost is $O(1)$ as
rebalancing rarely goes all the way up the tree. The worst-case cost is $O(\log
n)$, bounded by the tree height.

\begin{figure}
  \centerline{\includegraphics[width=\columnwidth]{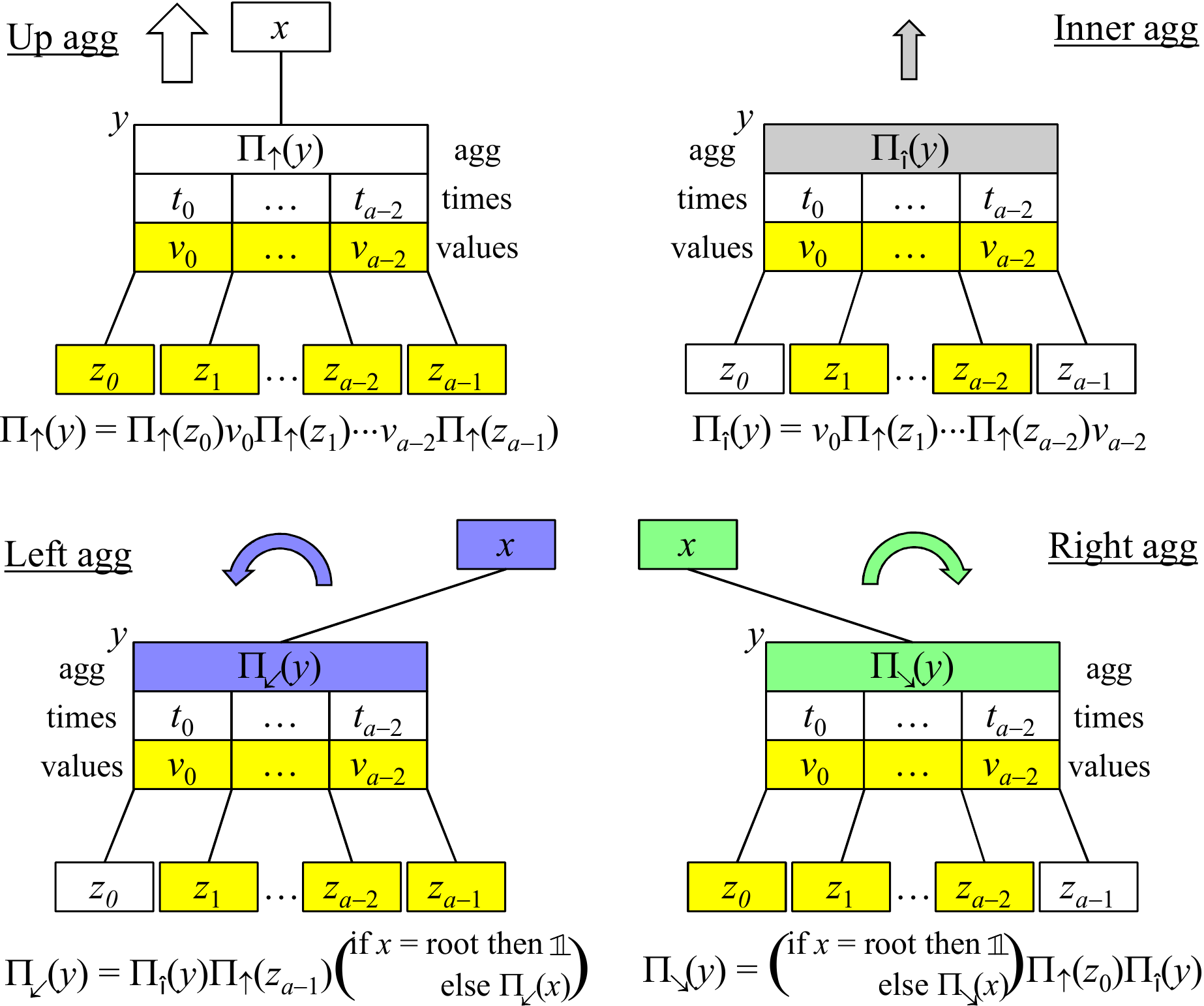}}
  \vspace*{-3mm}
  \caption{\label{fig_partials}Partial aggregates definitions.}
\end{figure}

\smallskip

\noindent{}\textbf{Sub-$O(\log n)$ Repair.}
The basic algorithm stores at each node~$y$ the up-aggregate $\Pi_\uparrow(y)$,
i.e., the partial aggregate of the subtree under~$y$. This is problematic,
because it means that an
insertion or eviction at a node~$z$, usually a leaf, affects the partial
aggregates stored in all ancestors of $z$---that is, the entire path up to the
root. To circumvent this issue, we need an arrangement of aggregates that can
be repaired by traversing to a finger, \emph{without} always traversing to the
root. For this, we make each node store the kind of partial aggregate suitable
for its position in the tree. Furthermore, because the root no longer contains
the aggregate of the whole tree, we will ensure that \lstinline{query()} can be
answered by combining partial aggregates at the left finger, the root, and the
right finger.

To meet these requirements, we define four kinds of partial aggregates in
Figure~\ref{fig_partials}. As illustrated in Figure~\ref{fig_finger}, they are
used in a B-tree according to the following \emph{aggregation invariants}:
\begin{itemize}[leftmargin=0em,label={},itemsep=2pt]
\item $\rhd$\textbf{{Non-spine nodes store the up-aggregate~$\Pi_\uparrow$.}}
  Such a node is neither a finger nor an ancestor of a finger. This aggregate
  must be repaired whenever the subtree below it changes.
  Figure~\ref{fig_finger}(A) shows nodes with up-aggregates in white, light
  blue, or light green. For example, the center child of the root contains the
  aggregate \smallsf{hijklmn}, comprising its entire subtree.

\item $\rhd$\textbf{{The root stores the inner
      aggregate~$\Pi_{\hat{\scriptscriptstyle |}}$.}}
  This aggregate is only affected by changes to the inner part of the tree, and
  not by changes below the left-most or right-most child of the root.
  Figure~\ref{fig_finger}(A) shows the inner parts of the tree in white and the
  root in gray, and the root stores the aggregate \smallsf{ghijklmno}.

\item $\rhd$\textbf{{Non-root nodes on the left spine store the left
      aggregate~$\Pi_\swarrow$.}} For a given node~$y$, the left aggregate
  encompasses all nodes under the left-most child of the root except for $y$'s
  left-most child~$z_0$. When a change occurs below the left-most child of the
  root, the only aggregates that need to be repaired are those on a traversal up
  to the left spine and then down to the left finger. Figure~\ref{fig_finger}(A)
  shows the left spine in dark blue and nodes affecting it in light blue. For
  example, the node in the middle of the left spine contains the aggregate
  \smallsf{cdef}, comprising the left subtree of the root except for the left
  finger.

\item $\rhd$\textbf{{Non-root nodes on the right spine store the right
      aggregate~$\Pi_\searrow$}}. This is symmetric to the left
  aggregate~$\Pi_\swarrow$. When a change occurs below the right-most child of
  the root, only aggregates on a traversal to the right finger are repaired.
  Figure~\ref{fig_finger}(A) shows the right spine in dark green and nodes
  affecting it in light green. For example, the node in the middle of the right
  spine contains the aggregate \smallsf{qst} of the right subtree of the root
  except for the right finger.
\end{itemize}

\begin{figure}
  \includegraphics[scale=0.41]{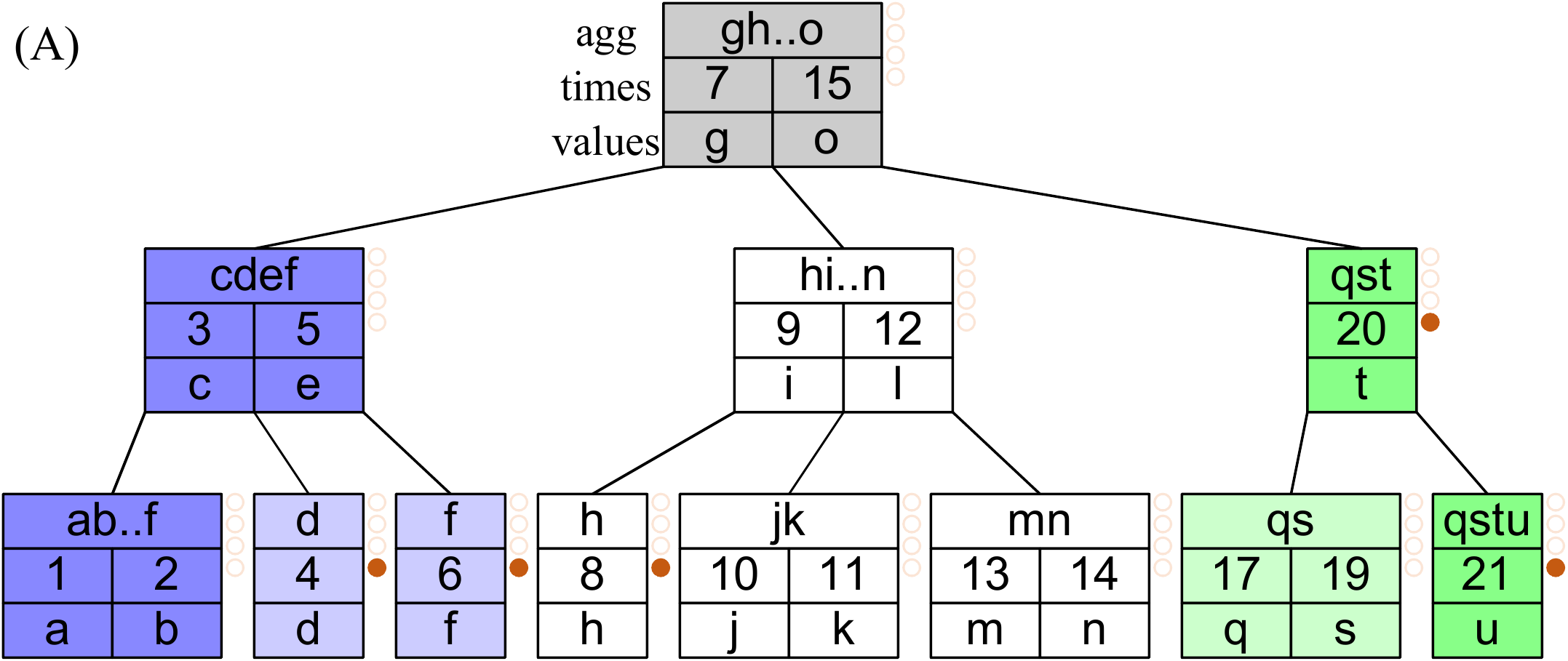}\\
  \noindent\mbox{\small Step A$\to$B, in-order insert \textsf{22:v}. Spent 0, refunded 1.}\\
  \includegraphics[scale=0.41]{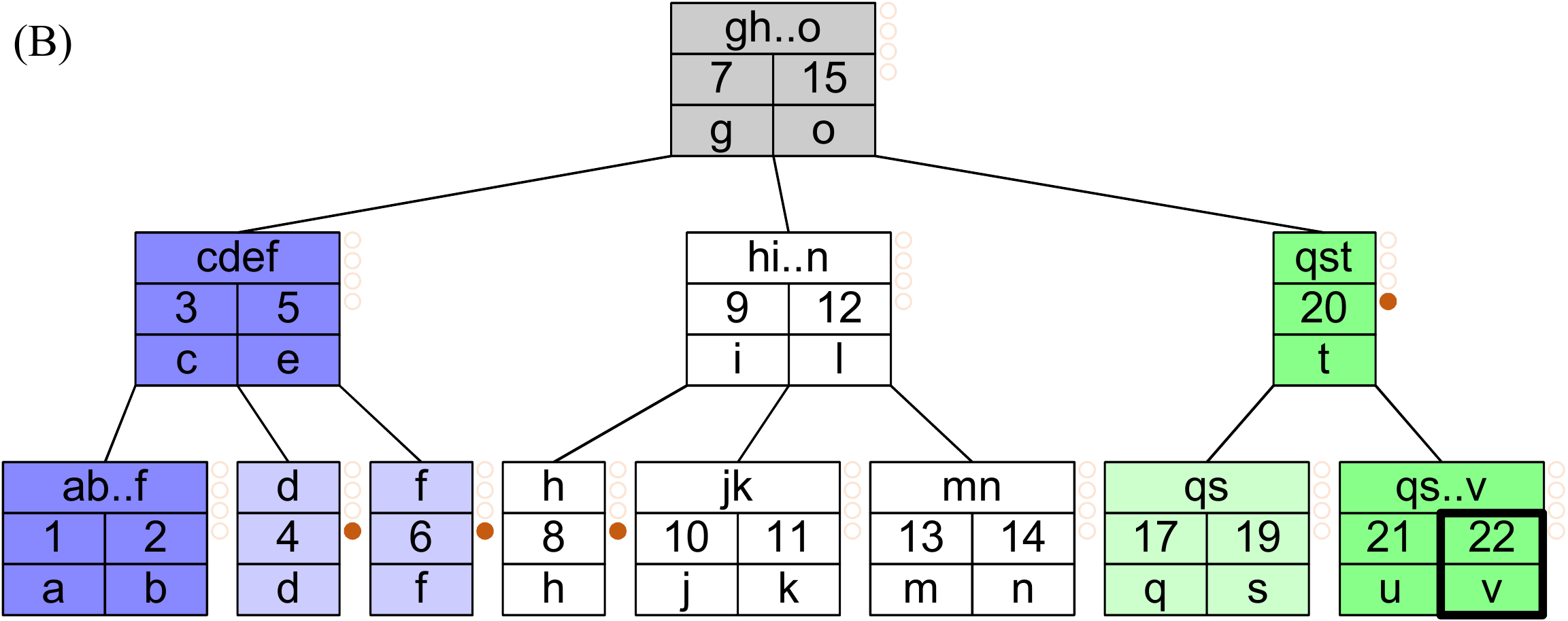}\\
  \noindent\mbox{\small Step B$\to$C, out-of-order insert \textsf{18:r}. Spent 0, billed 2.}\\
  \includegraphics[scale=0.41]{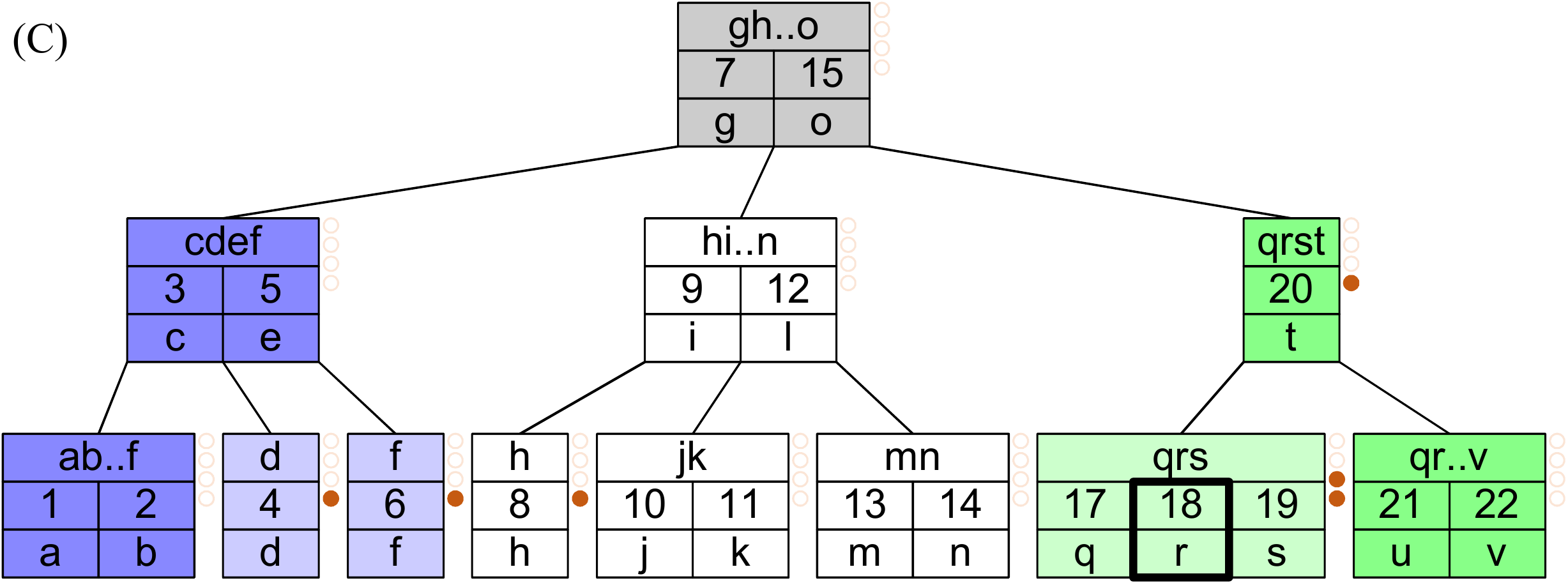}\\
  \noindent\mbox{\small Step C$\to$D, evict \textsf{1:a}. Spent 0, billed 1.}\\
  \includegraphics[scale=0.41]{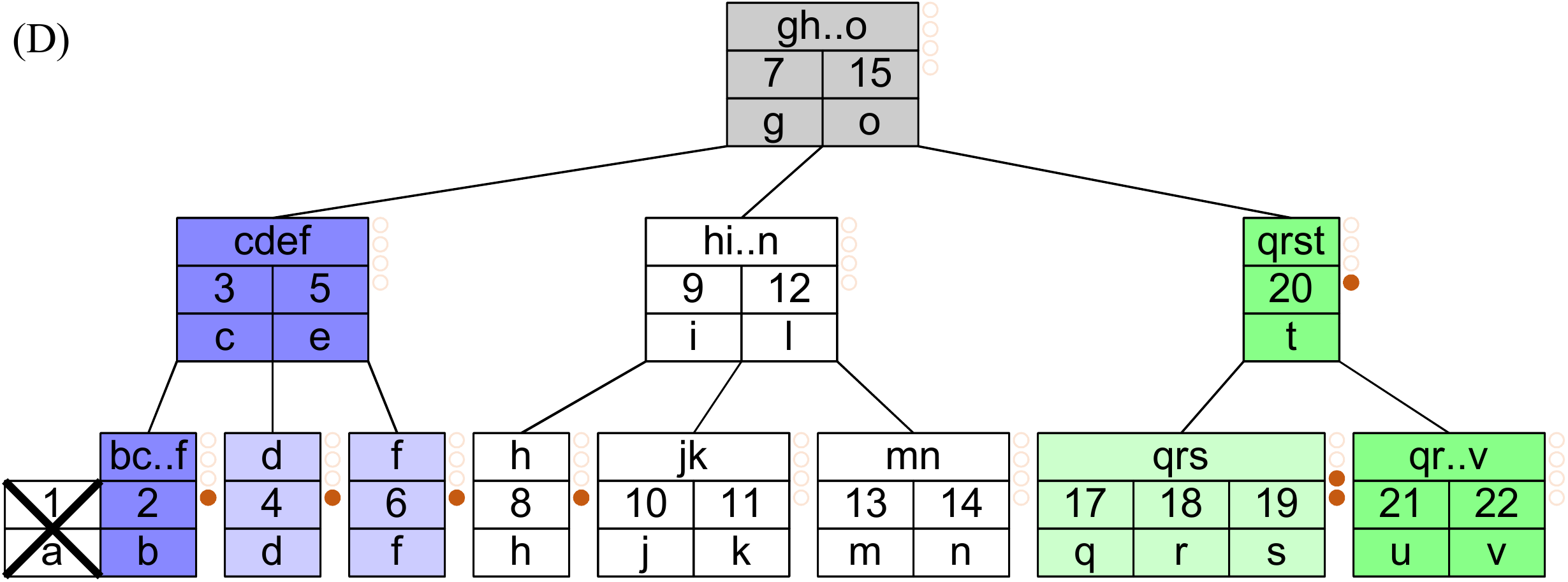}\\
  \noindent\mbox{\small Step D$\to$E, out-of-order insert \textsf{16:p}, split. Spent 1, refunded 1.}\\
  \includegraphics[scale=0.41]{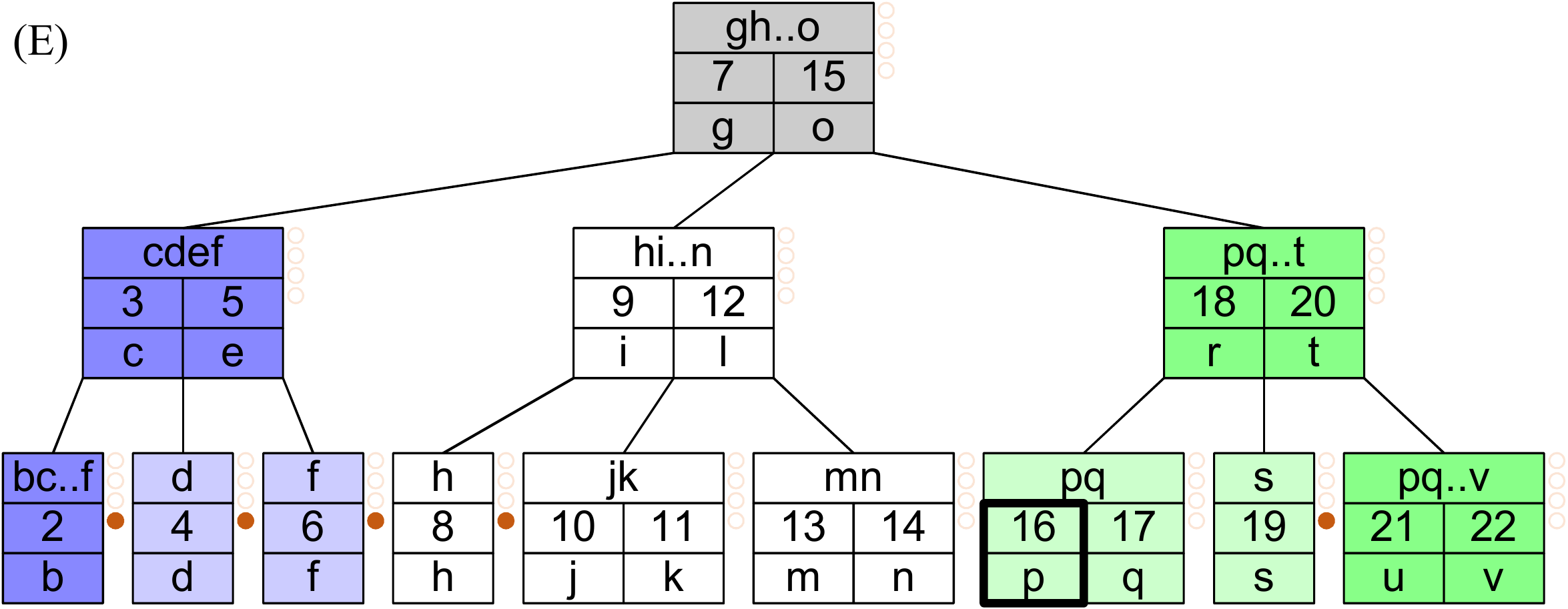}\\
  \noindent\mbox{\small Step E$\to$F, evict \textsf{2:b}, merge. Spent 1, billed 0.}\\
  \includegraphics[scale=0.41]{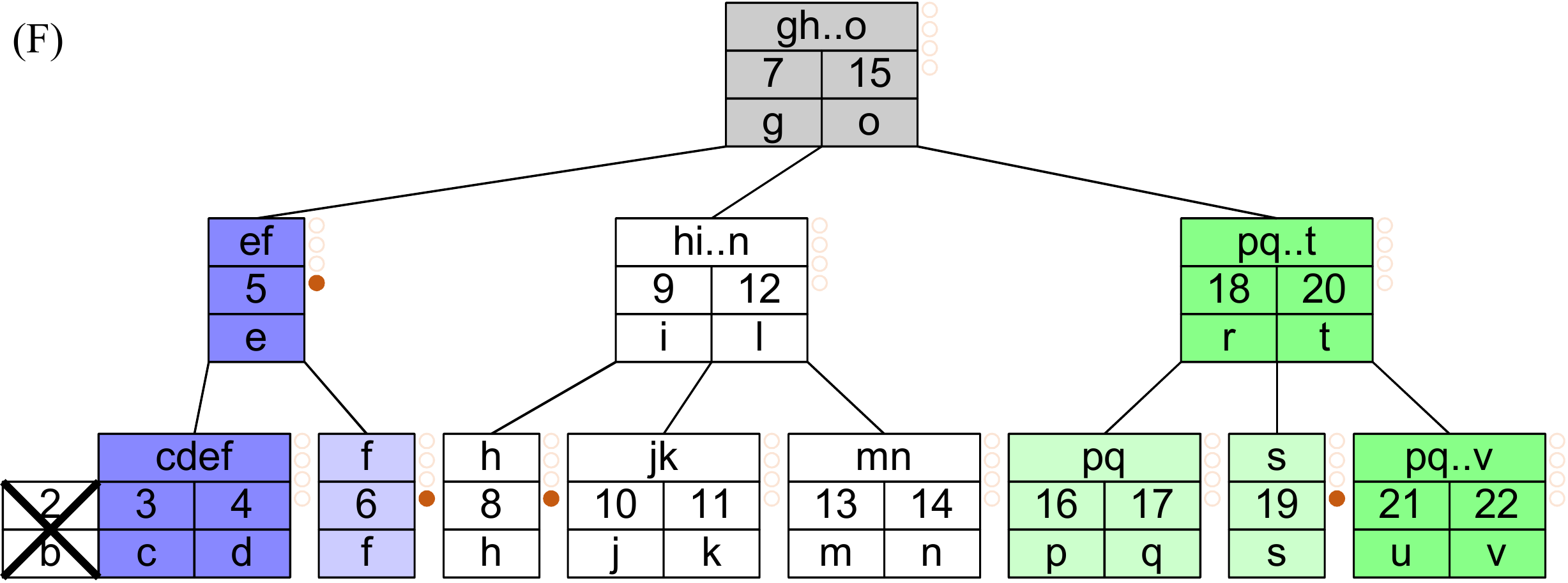}\\
  \vspace*{-3mm}
  \caption{\label{fig_finger}Finger B-tree with aggregates: example.}
\end{figure}

\begin{figure}
  \includegraphics[scale=0.41]{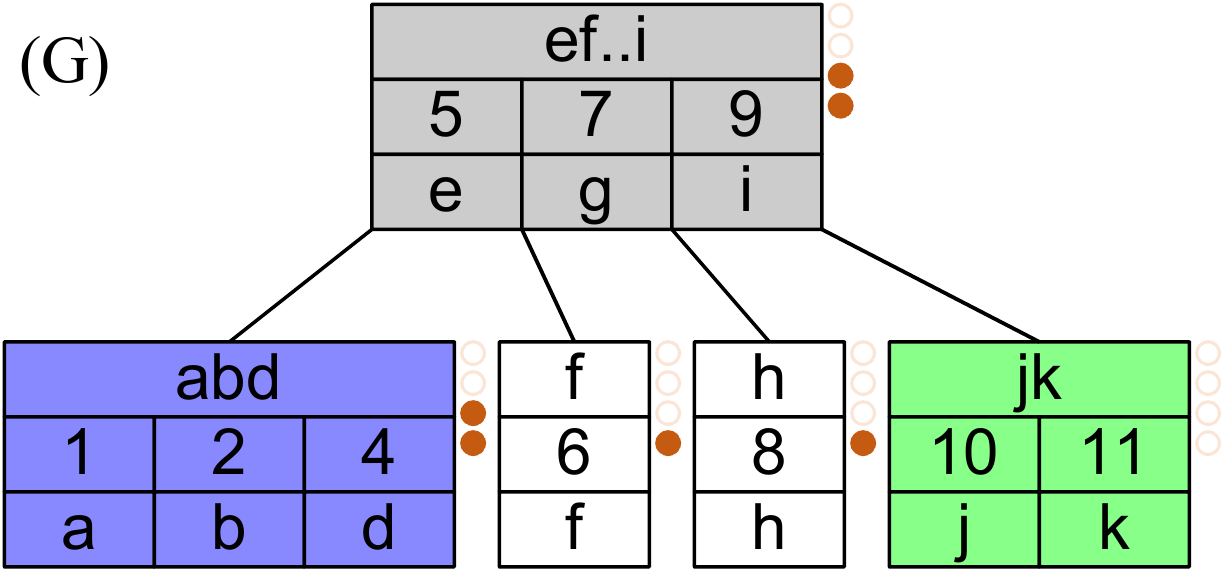}\\
  \noindent\mbox{\small Step G$\to$H, insert \textsf{3:c}, split, height increase and split.}\\
  \mbox{\small Spent 2, billed 0.}\\
  \includegraphics[scale=0.41]{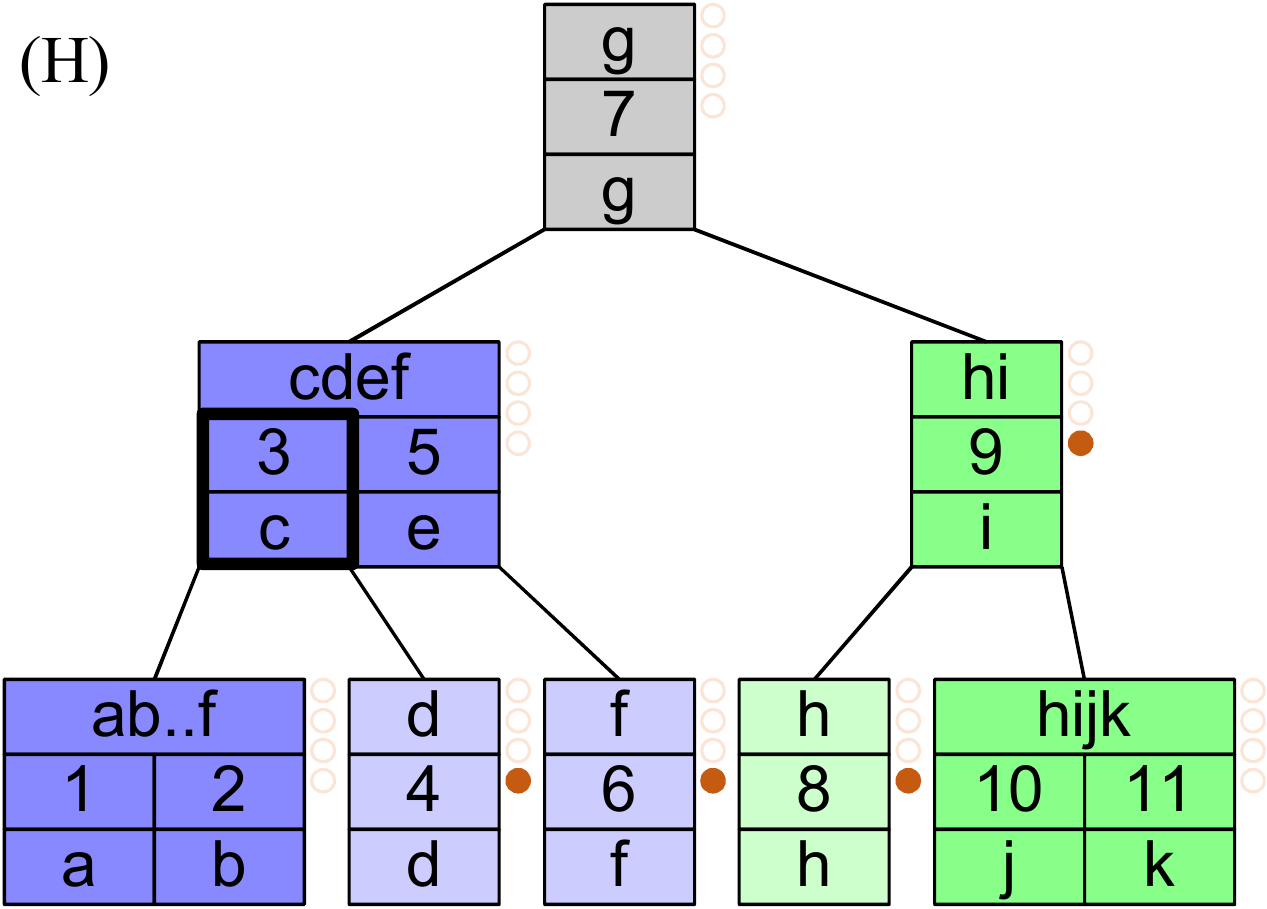}
  \vspace*{-3mm}
  \caption{\label{fig_inc_height}Finger B-tree height increase and split.}
  \vspace*{5mm}
  \includegraphics[scale=0.41]{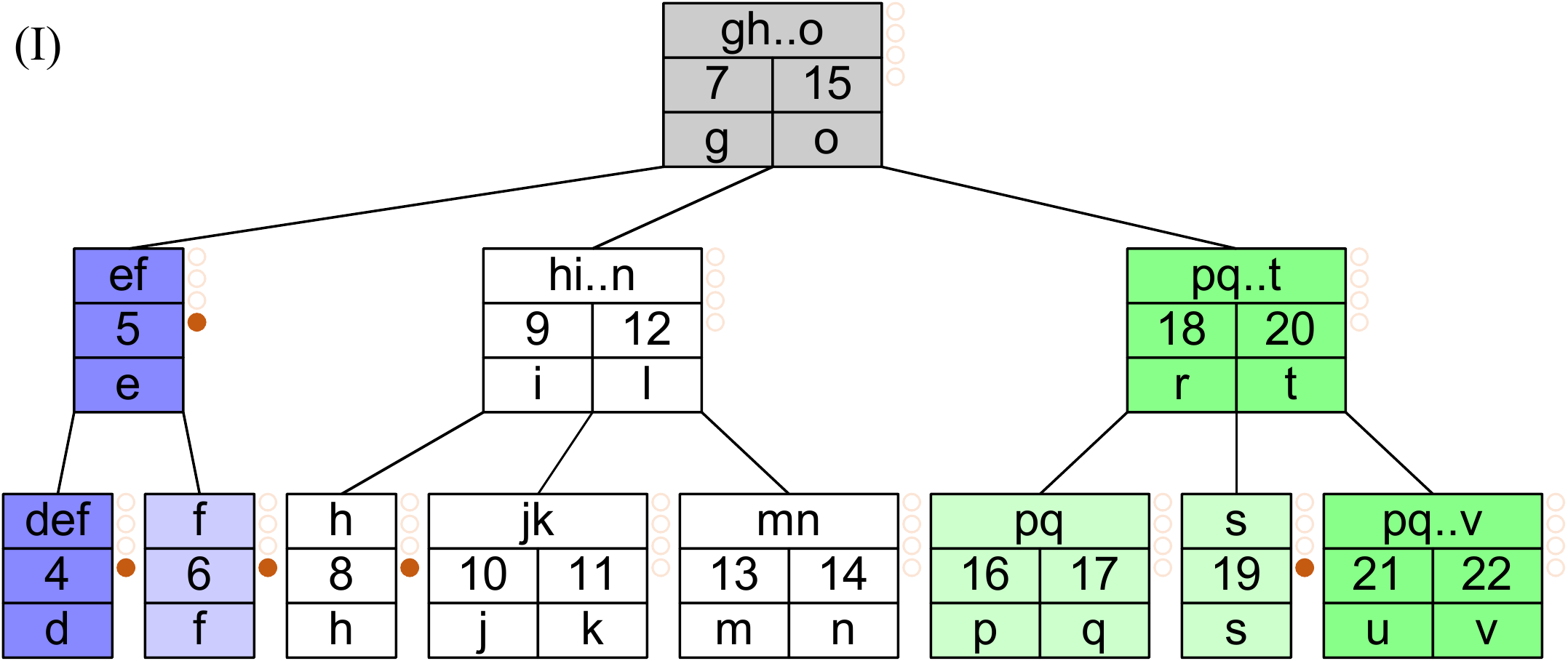}\\
  \noindent{\small Step I$\to$J, evict \textsf{4:d}, merge, move. Spent 2, billed 1.}\\
  \includegraphics[scale=0.41]{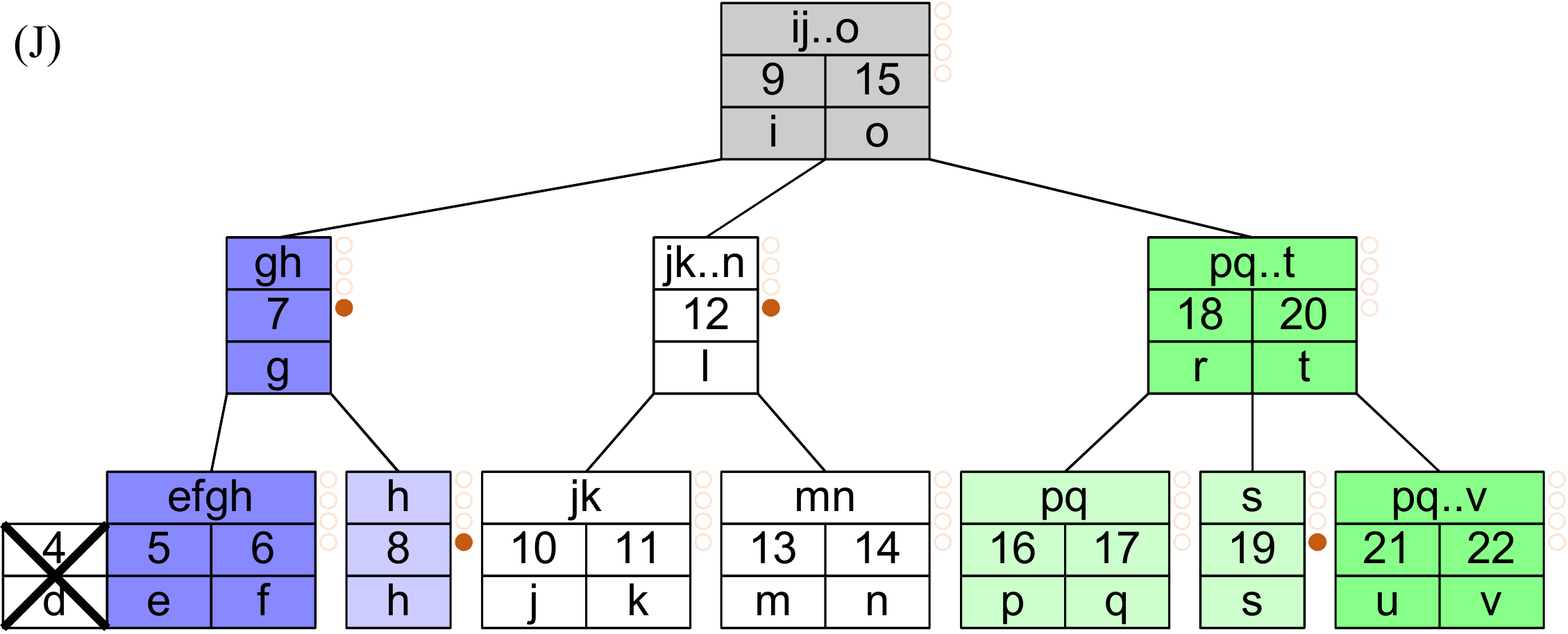}\\
  \vspace*{-3mm}
  \caption{\label{fig_move}Finger B-tree move.}
  \vspace*{5mm}
  \includegraphics[scale=0.41]{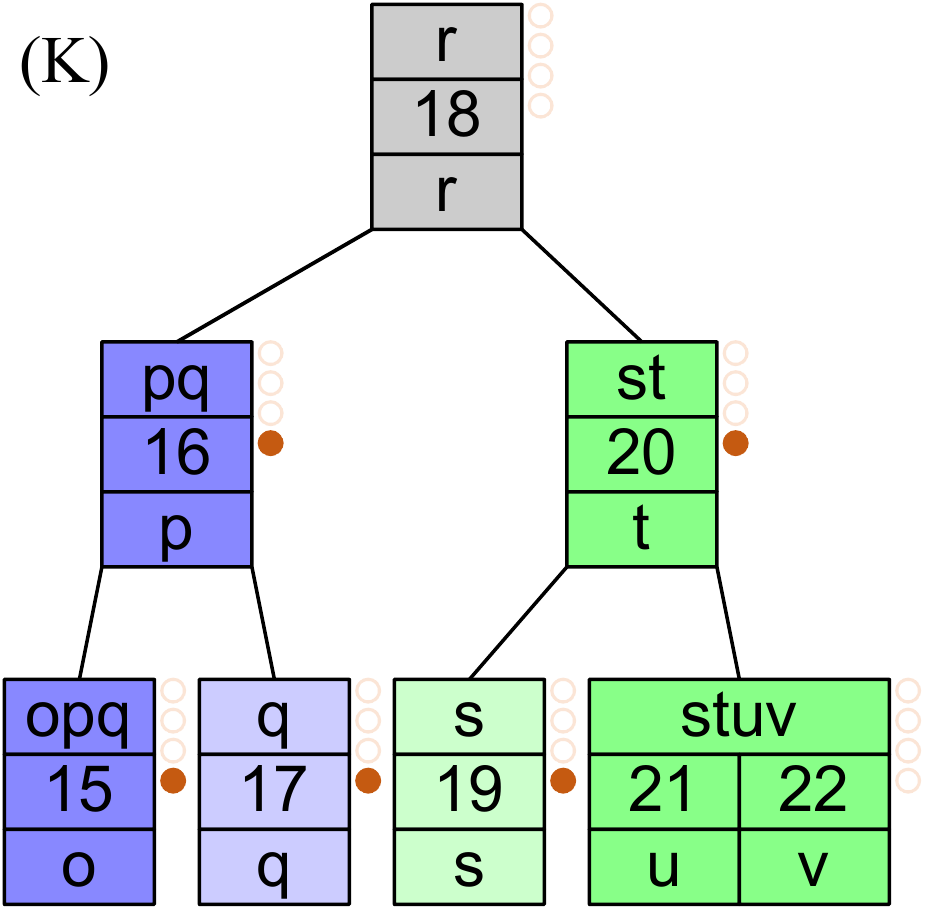}\\
  \noindent\mbox{\small Step K$\to$L, evict \textsf{15:o}, merge, merge and height decrease.}\\
  \noindent\mbox{\small Spent 2, refunded 2.}\\
  \includegraphics[scale=0.41]{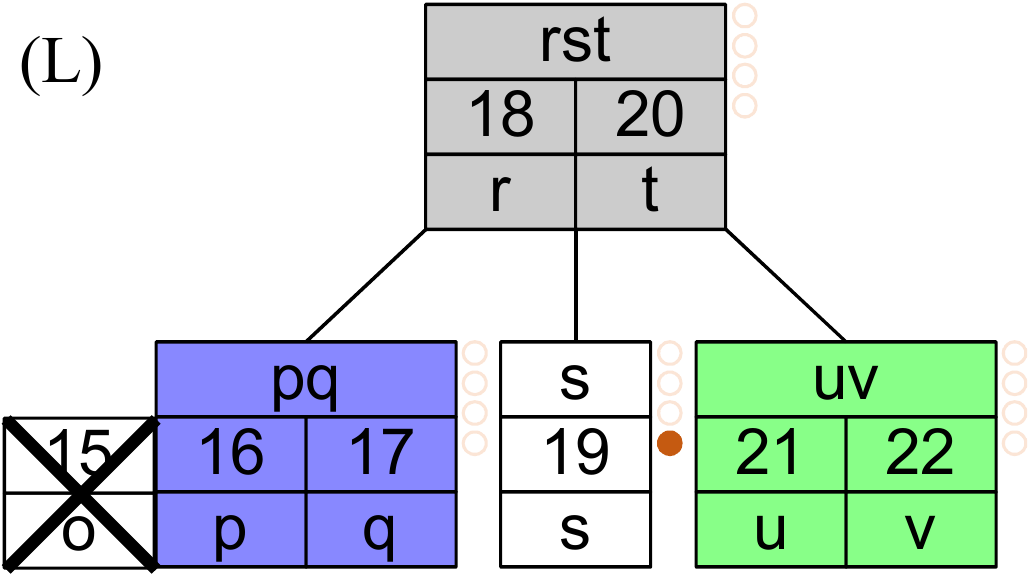}\\
  \vspace*{-3mm}
  \caption{\label{fig_dec_height}Finger B-tree merge and height decrease.}
\end{figure}

\subsection{Using Finger B-Trees}
\label{sec:algorithm}

This section describes an algorithm that implements the OoO SWAG using a finger
B-tree augmented with aggregates. It achieves sub-$O(\log n)$ time complexity by
maintaining the size invariants from Section~\ref{sec:classic} and the
aggregation invariants from Section~\ref{sec:intuition}.

The algorithmic complexity analysis will account for the cost of
split, merge, or move operations by counting \emph{coins}.
Specifically, the analysis counts the number of split, merge, or move
steps of an insert or evict operation as \emph{spent} coins.  Coins
can be imagined as being stored at tree nodes, so they can be used to
pay for split, merge, or move operations later.  Throughout this
paper, coins are visualized as little golden circles next to tree
nodes. Sometimes, coins must be added or removed from the outside to
make up the difference between spent coins and coins in the tree
before and after each step. We refer to these coins as being
\emph{billed} or \emph{refunded}. The key result of the proof will be
that billed coins never exceed~2 for any \lstinline{insert($t,v$)} or
\lstinline{evict($t$)}, hence rebalancing has amortized constant
time complexity.

Figures~\ref{fig_finger}--\ref{fig_dec_height} show concrete examples
covering all the interesting cases of the algorithm. Each state, for
instance~(A), shows a tree with aggregates and coins. Each step, for
instance~A$\to$B, shows an insert or evict, illustrating how it affects the
tree, its partial aggregates, and coins.
\begin{itemize}[leftmargin=1em]
  \item In Figure~\ref{fig_finger}, Step~A$\to$B is an in-order insert
    without rebalance, which only affects the aggregate at a single
    node, the right finger.
  \item Step~B$\to$C is an out-of-order insert without rebalance,
    affecting aggregates on a walk to the right finger.
  \item Step~C$\to$D is an in-order evict without rebalance,
    affecting the aggregate at a single node, the left finger.
  \item Step~D$\to$E is an out-of-order insert to a node with arity
    $a=2\,\cdot\,$\lstinline{MIN_ARITY}, causing an
    overflow; rebalancing splits it.
  \item Step~E$\to$F is an evict from a node with
    $a\,=\,$\lstinline{MIN_ARITY}, causing the node to underflow;
    rebalancing merges it with its neighbor.
  \item In Figure~\ref{fig_inc_height}, Step~G$\to$H is an insert that
    causes nodes to overflow all the way up to the root, causing a
    height increase followed by splitting the old root. This affects
    aggregates on all split nodes and on both spines.
  \item In Figure~\ref{fig_move}, Step~I$\to$J is an evict that causes
    first an underflow that is fixed by a merge, and then an underflow
    at the next level where the neighbor node is too big to merge.
    The algorithm repairs the size invariant with a move of a child
    and a timed value from the neighbor.  This step affects aggregates
    on all nodes affected by rebalancing plus a walk to the left
    finger.
  \item In Figure~\ref{fig_dec_height}, Step~K$\to$L is an evict that
    causes nodes to underflow all the way up to the root, causing a
    height decrease to eliminate the old empty root. This affects
    aggregates on all merged nodes and on both spines.
\end{itemize}

\begin{figure*}[!t]
\begin{minipage}{\columnwidth}\begin{lstlisting}
fun query() : Agg
  if root.isLeaf()
    return root.agg
  return leftFinger.agg $\otimes$ root.agg $\otimes$ rightFinger.agg

fun insert($t$ : Time, $v$ : Agg)
  node $\gets$ searchNode($t$)
  node.localInsertTimeAndValue($t$, $v$)
  top, hit$_\texttt{left}$, hit$_\texttt{right}$ $\gets$ rebalanceForInsert(node)
  repairAggs(top, hit$_\texttt{left}$, hit$_\texttt{right}$)

fun evict($t$ : Time)
  node $\gets$ searchNode($t$)
  found, idx $\gets$ node.localSearch($t$)
  if found
    if node.isLeaf()
      node.localEvictTimeAndValue($t$)
      top,hit$_\texttt{left}$,hit$_\texttt{right}$ $\gets$ rebalanceForEvict(node, null)
    else
      top,hit$_\texttt{left}$,hit$_\texttt{right}$ $\gets$ evictInner(node, idx)
    repairAggs(top, hit$_\texttt{left}$, hit$_\texttt{right}$)

fun repairAggs(top : Node, hit$_\texttt{left}$ : Bool, hit$_\texttt{right}$ : Bool)
  if top.hasAggUp()
    while top.hasAggUp()
      top $\gets$ top.parent
      top.localRepairAgg()
  else
    top.localRepairAgg()
  if top.leftSpine or top.isRoot() and hit$_\texttt{left}$
    left $\gets$ top
    while not left.isLeaf()
      left $\gets$ left.getChild(0)
      left.localRepairAgg()
  if top.rightSpine or top.isRoot() and hit$_\texttt{right}$
    right $\gets$ top
    while not right.isLeaf()
      right $\gets$ right.getChild(right.arity - 1)
      right.localRepairAgg()
\end{lstlisting}\end{minipage}\hspace*{\columnsep}\begin{minipage}{\columnwidth}\begin{lstlisting}[firstnumber=last]
fun rebalanceForInsert(node : Node) : Node$\times$Bool$\times$Bool
  hit$_\texttt{left}$, hit$_\texttt{right}$ $\gets$ node.leftSpine, node.rightSpine
  while node.arity > MAX_ARITY
    if node.isRoot()
      heightIncrease()
      hit$_\texttt{left}$, hit$_\texttt{right}$ $\gets$ true, true
    split(node)
    node $\gets$ node.parent
    hit$_\texttt{left}$ $\gets$ hit$_\texttt{left}$ or node.leftSpine
    hit$_\texttt{right}$ $\gets$ hit$_\texttt{right}$ or node.rightSpine
  return node, hit$_\texttt{left}$, hit$_\texttt{right}$

fun rebalanceForEvict(node : Node, toRepair : Node)
                     : Node$\times$Bool$\times$Bool
  hit$_\texttt{left}$, hit$_\texttt{right}$ $\gets$ node.leftSpine, node.rightSpine
  if node $=$ toRepair
    node.localRepairAggIfUp()
  while not node.isRoot() and node.arity < MIN_ARITY
    parent $\gets$ node.parent
    nodeIdx, siblingIdx $\gets$ pickEvictionSibling(node)
    sibling $\gets$ parent.getChild(siblingIdx)
    hit$_\texttt{right}$ $\gets$ hit$_\texttt{right}$ or sibling.rightSpine
    if sibling.arity $\leq$ MIN_ARITY
      node $\gets$ merge(parent, nodeIdx, siblingIdx)
      if parent.isRoot() and parent.arity $=$ 1
        heightDecrease()
      else
        node $\gets$ parent
    else
      move(parent, nodeIdx, siblingIdx)
      node $\gets$ parent
    if node $=$ toRepair
      node.localRepairAggIfUp()
    hit$_\texttt{left}$ $\gets$ hit$_\texttt{left}$ or node.leftSpine
    hit$_\texttt{right}$ $\gets$ hit$_\texttt{right}$ or node.rightSpine
  return node, hit$_\texttt{left}$, hit$_\texttt{right}$
\end{lstlisting}\end{minipage}
\caption{\label{fig_algo}Finger B-Tree with aggregates: algorithm.}
\end{figure*}

\begin{figure*}
  \small
\begin{minipage}{\columnwidth}\begin{lstlisting}[firstnumber=last]
fun evictInner(node : Node, idx : Int) : Node$\times$Bool$\times$Bool
  left, right $\gets$ node.getChild(idx), node.getChild(idx+1)
  if right.arity > MIN_ARITY
    leaf, $t_\texttt{leaf}$, $v_\texttt{leaf}$ $\gets$ oldest(right)
  else
    leaf, $t_\texttt{leaf}$, $v_\texttt{leaf}$ $\gets$ youngest(left)
  leaf.localEvictTimeAndValue($t_\texttt{leaf}$)
  node.setTimeAndValue(idx, $t_\texttt{leaf}$, $v_\texttt{leaf}$)
  top,hit$_\texttt{left}$,hit$_\texttt{right}$ $\gets$ rebalanceForEvict(leaf, node)
  if top.isDescendent(node)
    while top $\neq$ node
      top $\gets$ top.parent
      hit$_\texttt{left}$ $\gets$ hit$_\texttt{left}$ or top.leftSpine
      hit$_\texttt{right}$ $\gets$ hit$_\texttt{right}$ or top.rightSpine
      top.localRepairAggIfUp()
  return top, hit$_\texttt{left}$, hit$_\texttt{right}$
\end{lstlisting}
\caption{\label{fig_evict_inner_algo}Finger B-Tree evict inner: algorithm.}
\end{minipage}\hspace*{\columnsep}\begin{minipage}{\columnwidth}
  \centering
  \includegraphics[scale=0.41]{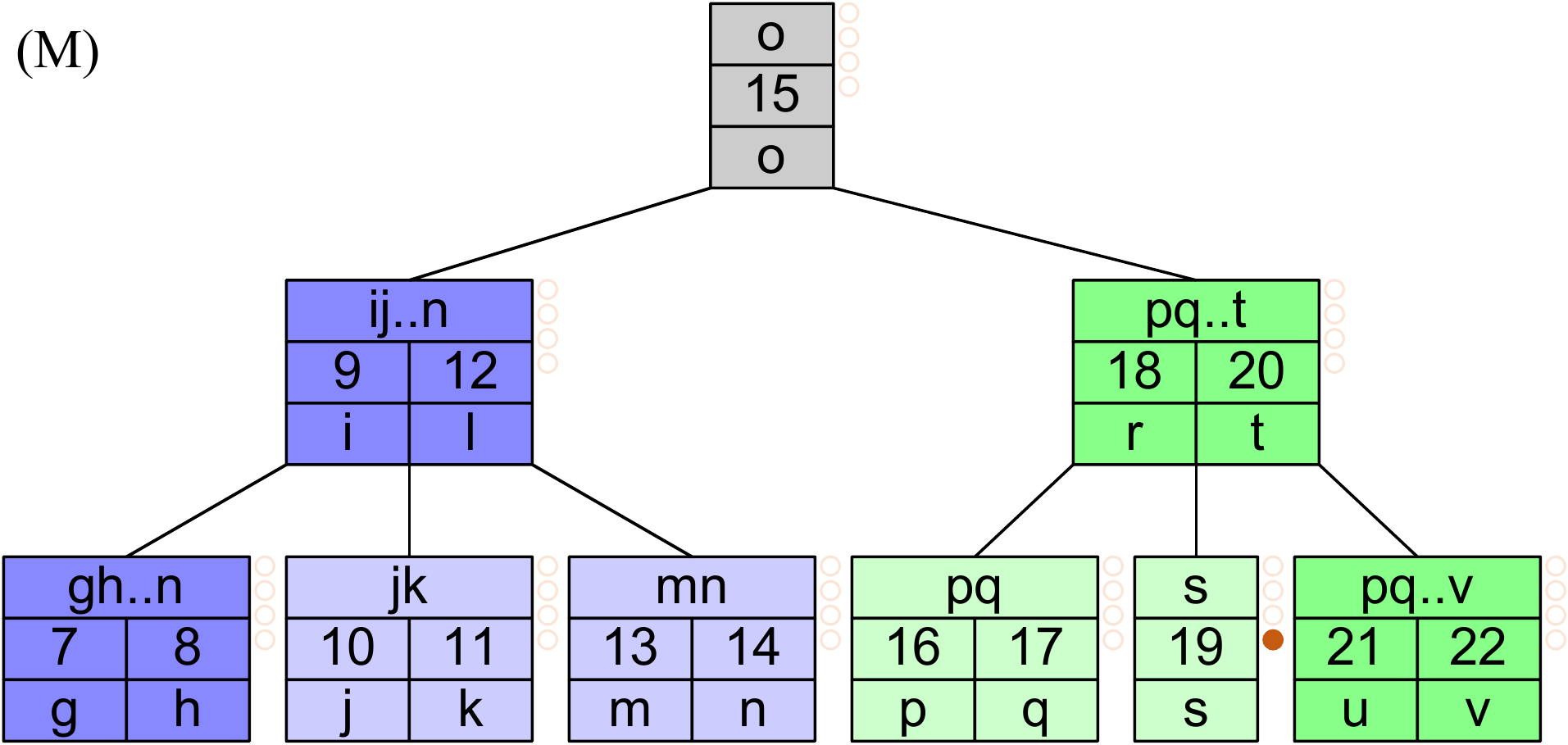}\\
  \noindent\mbox{\small Step M$\to$N, out-of-order evict \textsf{9:i}. Spent 0, billed 1.}\\
  \includegraphics[scale=0.41]{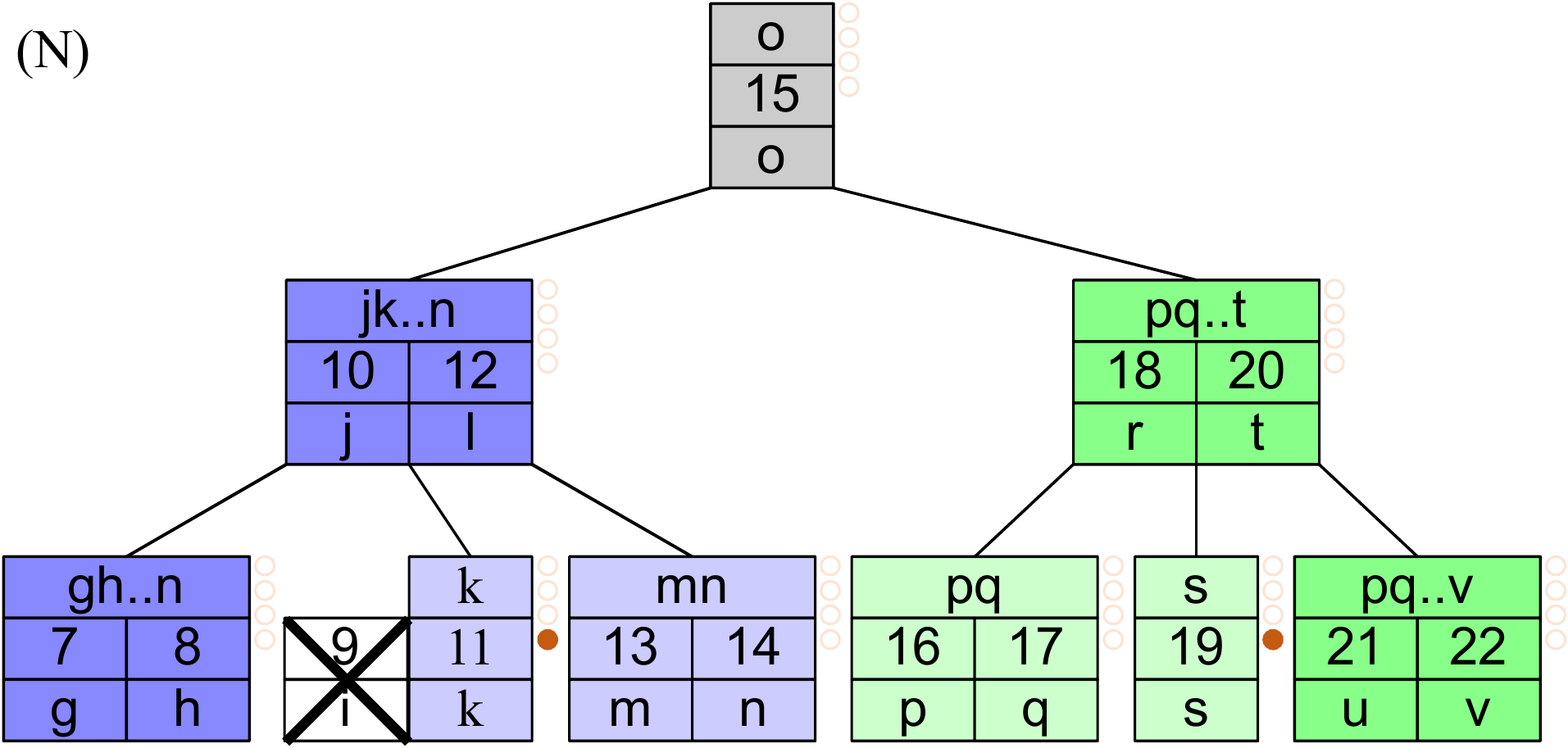}
  \vspace*{-3mm}
  \caption{\label{fig_evict_inner_example}Finger B-tree evict inner: example.}
\end{minipage}
\end{figure*}

Figure~\ref{fig_algo} shows most of the algorithm, excluding only
\lstinline{evictInner}, which will be presented later.  While
rebalancing always works bottom-up, aggregate repair works in the
direction of the partial aggregates: either up for up-agg or
inner-agg, or down for left-agg or right-agg.  Our algorithm
piggybacks the repair of up-aggs onto the local insert or evict and
onto rebalancing, and then repairs the remaining aggregates
separately.  To facilitate the handover from the piggybacked phase
to the dedicated phase of aggregate repair, the rebalancing routines
return a triple
\mbox{$\langle$\lstinline{top}, \lstinline{hit}$_\texttt{left}$, \lstinline{hit}$_\texttt{right}\rangle$}, for instance, in Line~9.
Node \lstinline{top} is where rebalancing topped out, and if it has an
up-agg, it is the last node whose aggregate has already been repaired.
Booleans \lstinline{hit}$_\texttt{left}$ and
\lstinline{hit}$_\texttt{right}$ indicate whether rebalancing affected
the left or right spine, determining whether aggregates on the
respective spine have to be repaired.

To keep the algorithm more readable, we factored out the case of
evicting from a non-leaf node into function \lstinline{evictInner} in
Figure~\ref{fig_evict_inner_algo}. To evict something from an inner
node, Line~82 evicts a substitute from a leaf instead, and Line~83
writes that substitute over the evicted slot.  Function
\lstinline{evictInner} creates an obligation to repair an extra node
during rebalancing, handled by parameter \lstinline{toRepair} on Line~52 in
the same figure.  Function \lstinline{evictInner} can
only be triggered for out-of-order eviction, because in-order
evictions always happen at the left finger, which is a leaf.

The following theorems state our correctness guarantees and the time
complexity; their proofs appear in Appendix~\ref{sec:appdxproof}.

\begin{theorem}\label{trm_finger_correctness}
  In a finger B-tree with aggregates that contains
  \mbox{$\tv{t_1}{v_1},\ldots,\tv{t_n}{v_n}$}, operation
  \lstinline{query()} returns \mbox{$v_1\otimes\ldots\otimes v_n$}.
\end{theorem}

\begin{theorem}\label{trm_finger_complexity}
  In a finger B-tree with aggregates, \lstinline{query()} costs at most~$O(1)$
  time, and \lstinline{insert($t,v$)} and \lstinline{evict($t$)} take time
  $T_\textrm{search}+T_\textrm{rebalance}+T_\textrm{repair}$, where
  \begin{itemize}[topsep=2pt]
    \item $T_\textrm{search}$ is $O(\log d)$, with $d$ being the
      distance to the start or end of the window, whichever is closer;
    \item $T_\textrm{rebalance}$ is amortized $O(1)$ and worst-case
      $O(\log n)$; and
    \item $T_\textrm{repair}$ is $O(T_\textrm{search} + T_\textrm{rebalance})$.
  \end{itemize}
\end{theorem}


%% file: sharing.tex
\section{Window Sharing}
\label{sec:sharing}

\input{code-sharing}
This section explains how to use a single finger B-tree to efficiently answer
aggregations on subwindows of different sizes on the fly. Applications are
numerous.
One common basic example is a simple anomaly detection workflow that
compares two related aggregations: one on a large window representing the normal
``stable'' behavior and the other on a smaller window representing the most
recent behavior. Then, an alert is triggered when the aggregates differ
substantially. Whereas in this example, the sizes of the windows are known ahead
of query time, in many other applications---e.g., interactive data
exploration---queries are ad hoc.

We propose to implement window sharing via range que\-ries, as defined
at the end of Section~\ref{sec:problem}. This has many
benefits: The window contents need to be saved only once regardless of how many
subwindows are involved. Thus, each insert or evict needs to be performed only
once on the largest window. This approach can accommodate an arbitrary number of
shared window sizes. For instance, many users can register queries over
different window sizes. Importantly, queries can be ad hoc and interactive,
which would otherwise not be possible to support using multiple fixed instances.
Furthermore, the range-query formulation also accommodates the case where the
window boundary is not the current time (\mbox{$t_\texttt{to}\neq
  t_\texttt{now}$}). For instance, it can report results with some time-lag
dictated by punctuation or low watermarks.

To answer the range query
\mbox{\lstinline{query(}$t_\texttt{from}, t_\texttt{to}$\lstinline{)}}, the
algorithm, shown in Figure~\ref{fig_range_query}, uses recursion starting from
the least-common ancestor node whose subtree encompasses the queried range. The
main technical challenge is to avoid making spurious recursive calls. Because
the nodes already store partial aggregates, the algorithm should only recurse
into a node's children if the partial aggregates cannot be used directly.
Specifically, we aim for the algorithm to invoke at most two chains of recursive
calls, one visiting ancestors of \lstinline{node}$_\texttt{from}$ and the other
visiting ancestors of \lstinline{node}$_\texttt{to}$. The insight for preventing
spurious recursive calls is that one needs information about neighboring
timestamps in a node's parent to determine whether the node itself is subsumed
by the range. We encode whether the neighboring timestamp in the parent is
included in the range on the left or right by using $-\infty$ or $+\infty$,
respectively.

This strategy alone would have been similar to range query in an interval
tree~\cite{cormen_leiserson_rivest_1990}, albeit without explicitly storing the
ranges; however, our specially-designed partial aggregates add another layer of
details: not all nodes store agg-up values~$\Pi_\uparrow(y)$. But any nodes that
lack $\Pi_\uparrow(y)$ are guaranteed to be on one of the two recursion chains,
because if a query involves spines of the entire window, then those spines
coincide with edges of the intersection between the window and the range.

\begin{theorem}\label{trm_range_correctness}
  In a finger B-tree with aggregates that contains
  \mbox{$\tv{t_1}{v_1},\ldots,\tv{t_n}{v_n}$}, the operation
  \mbox{\lstinline{query(}$t_\texttt{from},t_\texttt{to}$\lstinline{)}} returns
  the aggregate \mbox{$v_{i_\texttt{from}}\otimes\ldots\otimes
    v_{i_\texttt{to}}$}, where $i_\texttt{from}$ is the largest $i$ such that
  $t_\texttt{from}\le t_{i_\texttt{from}}$ and $i_\texttt{to}$ is the smallest
  $i$ such that $t_{i_\texttt{to}}\le t_\texttt{to}$.
\end{theorem}

\begin{proof}
  By induction. Each recursive call returns the aggregate of the intersection
  between its subtree and the queried range.
\end{proof}

\begin{theorem}\label{trm_range_complexity}
  In a finger B-tree with aggregates that contains
  \mbox{$\tv{t_1}{v_1},\ldots,\tv{t_n}{v_n}$}, the operation
  \mbox{\lstinline{query(}$t_\texttt{from},t_\texttt{to}$\lstinline{)}}
  takes time \mbox{$O(\log d_\texttt{from} + \log d_\texttt{to} + \log n_\texttt{sub})$}, where
  \begin{itemize}[topsep=0pt]
    \item $i_\texttt{from}$ is the largest index $i$ such that
      $t_\texttt{from}\le t_{i_\texttt{from}}$
    \item $i_\texttt{to}$ is the smallest index $i$ such that
      $t_{i_\texttt{to}}\le t_\texttt{to}$
    \item $d_\texttt{from} = \min(i_\texttt{from}, n-i_\texttt{from})$ and
      $d_\texttt{to} = \min(i_\texttt{to}, n-i_\texttt{to})$ are the distances
      to the window boundary
    \item $n_\texttt{sub} = i_\texttt{to} - i_\texttt{from}$ is the size
      of subwindow being queried.
  \end{itemize}
\end{theorem}

\begin{proof}
  Using finger searches, Line~2 takes $O(\log d_\texttt{from} + \log
  d_\texttt{to})$. Now the distance from either \lstinline{node}$_\texttt{from}$
  or \lstinline{node}$_\texttt{to}$ to the least-common ancestor (LCA) is at
  most $O(\log n_\texttt{sub})$. Therefore, locating the LCA takes at most
  $O(\log n_\texttt{sub})$, and so do subsequent recursive calls in
  \lstinline{queryRec} that traverse the same paths.
\end{proof}

In particular, when a query ends at the current time (i.e., when $t_\texttt{\itshape
  to} = t_\texttt{\itshape now}$), the theorem says that the query takes $O(\log
n_\texttt{\itshape sub})$ time, where $n_\texttt{\itshape sub}$ is the size of
the subwindow being queried.


%% file: code-sharing.tex
\begin{figure}[!t]
\begin{lstlisting}[firstnumber=last]
fun query($t_\texttt{from}$ : Time, $t_\texttt{to}$ : Time) : Agg
  node$_\texttt{from}$, node$_\texttt{to}$ $\gets$ searchNode($t_\texttt{from}$), searchNode($t_\texttt{to}$)
  node$_\texttt{top}$ $\gets$ leastCommonAncestor(node$_\texttt{from}$, node$_\texttt{to}$)
  return queryRec(node$_\texttt{top}$, $t_\texttt{from}$, $t_\texttt{to}$)

fun queryRec(node : Node, $t_\texttt{from}$ : Time, $t_\texttt{to}$ : Time) : Agg
  if $t_\texttt{from}=-\infty$ and $t_\texttt{to}=+\infty$ and node.hasAggUp()
    return node.agg
  res $\gets$ $\onef{}$
  if not node.isLeaf()
    $t_\texttt{next}$ $\gets$ node.getTime(0)
    if $t_\texttt{from}<t_\texttt{next}$
      res = res $\otimes$ queryRec(node.getChild(0),
                           $\;t_\texttt{from}$,
                           $\;t_\texttt{next}\le t_\texttt{to}$ ? $+\infty$ : $t_\texttt{to})$
  for $i$ $\in$ [0, ..., node.arity - 2]
    $t_i$ $\gets$ node.getTime($i$)
    if $t_\texttt{from}\le t_i$ and $t_i\le t_\texttt{to}$
      res $\gets$ res $\otimes$ node.getValue(i)
    if not node.isLeaf() and $i+1<$ node.arity$\,\mathop{-}2$
      $t_{i+1}$ $\gets$ node.getTime($i+1$)
      if $t_i<t_\texttt{to}$ and $t_\texttt{from}<t_{i+1}$
        res $\gets$ res $\otimes$ queryRec(node.getChild($i+1$),
                              $\;t_\texttt{from}\le t_i$ ? $-\infty$ : $t_\texttt{from}$,
                              $\;t_{i+1}\le t_\texttt{to}$ ? $+\infty$ : $t_\texttt{to}$)
  if not node.isLeaf()
    $t_\texttt{curr}$ $\gets$ node.getTime(node.arity - 2)
    if $t_\texttt{curr}<t_\texttt{to}$
      res = res $\otimes$ queryRec(node.getChild(node.arity - 1),
                           $\;t_\texttt{from}\le t_\texttt{curr}$ ? $-\infty$ : $t_\texttt{from}$,
                           $\;t_\texttt{to}$)
  return res
\end{lstlisting}
\caption{\label{fig_range_query}Range query algorithm.}
\end{figure}

%% file: results.tex
\section{Results}
\label{sec:results}



\begin{figure*}[!t]
\center
\includegraphics[width=0.32\textwidth]{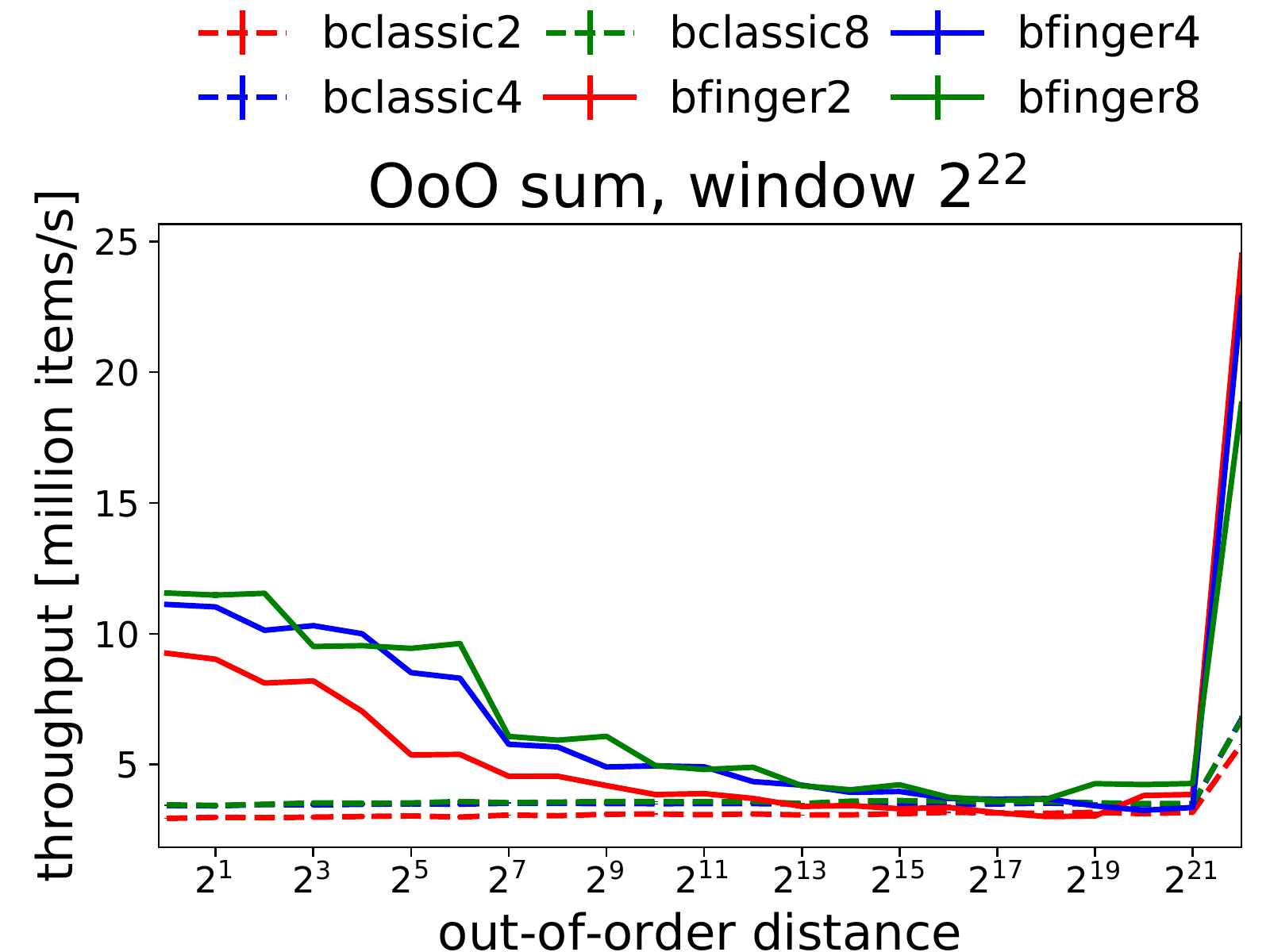}
\includegraphics[width=0.32\textwidth]{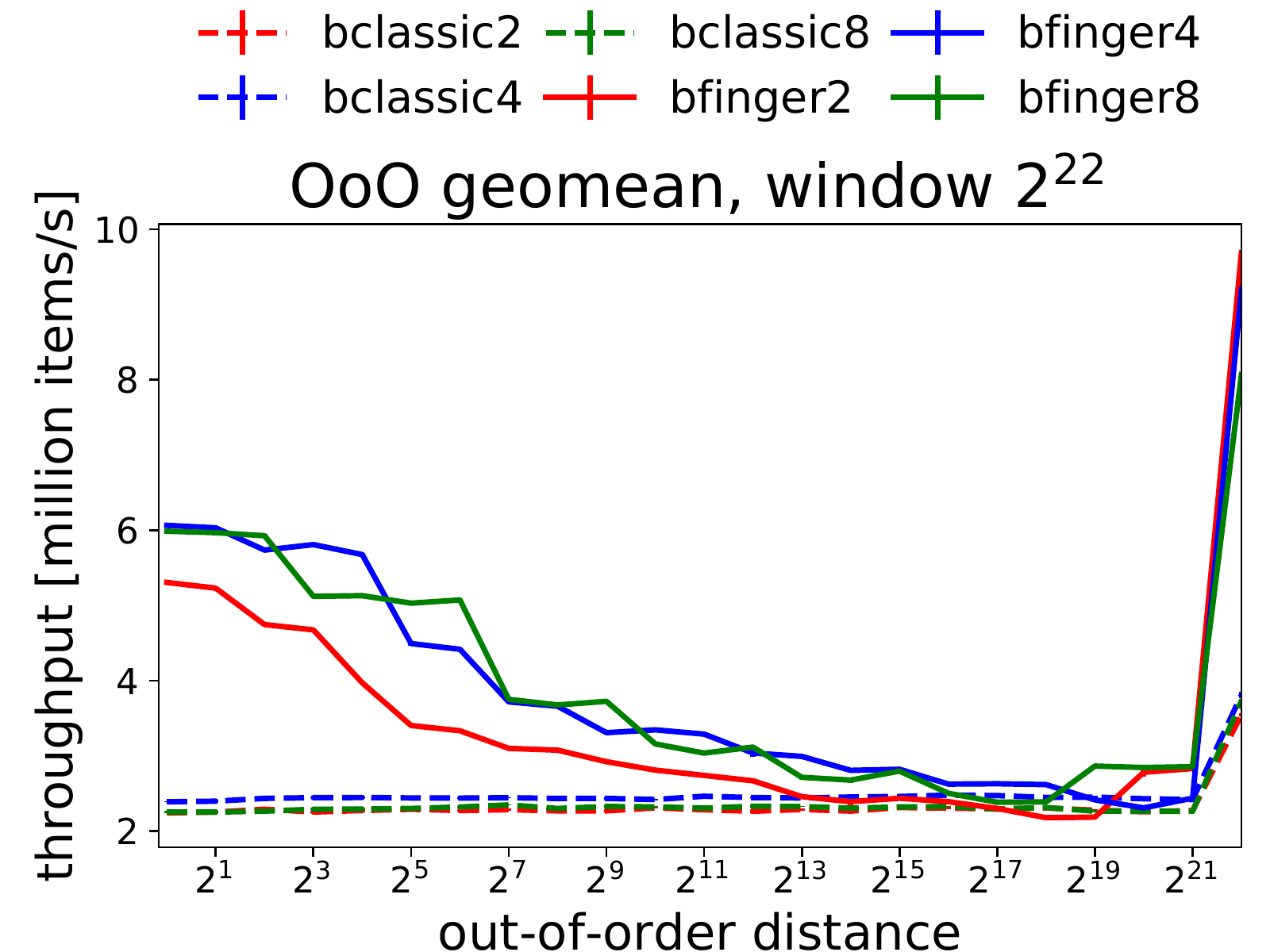}
\includegraphics[width=0.32\textwidth]{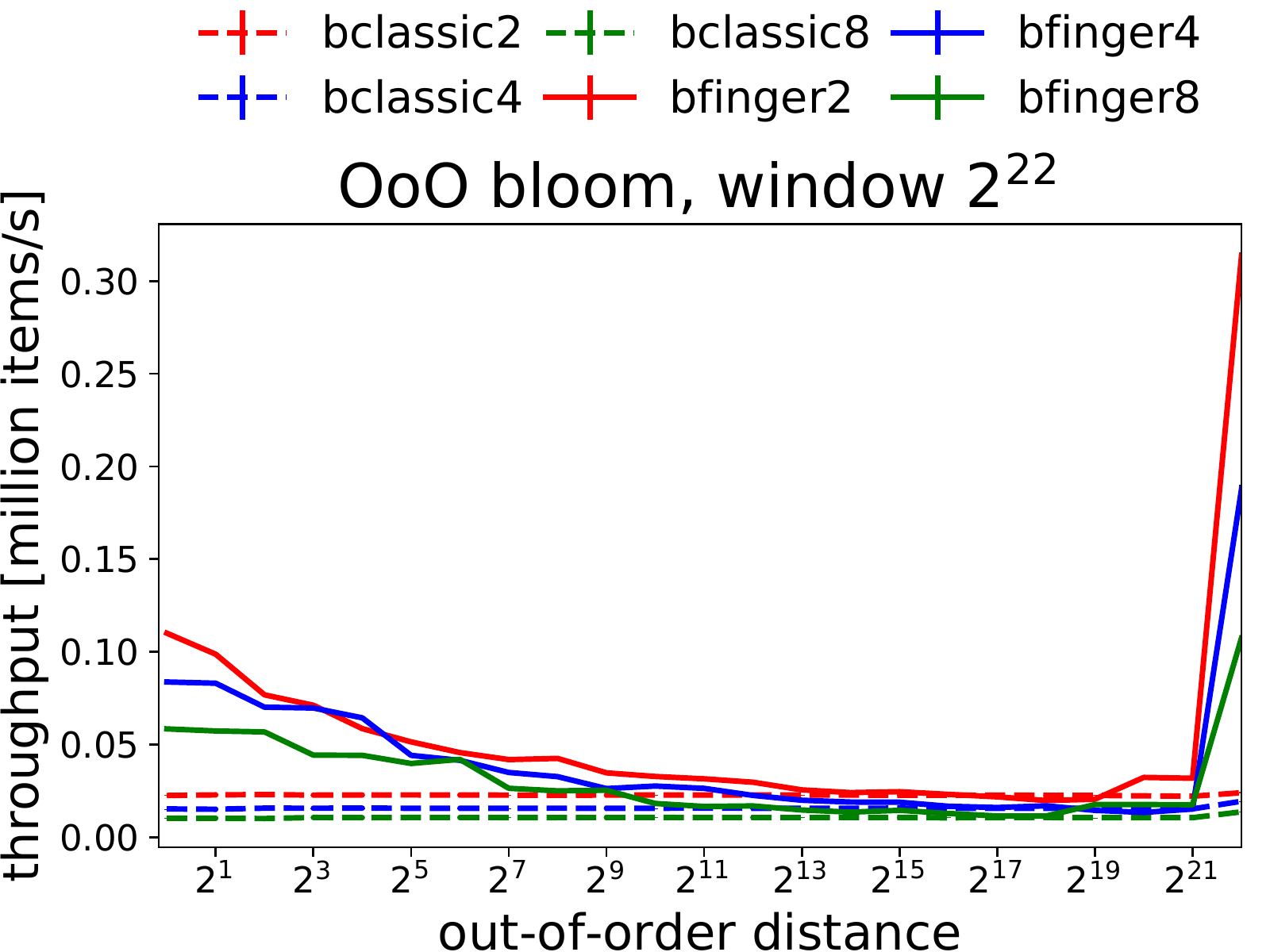}
\caption{Out-of-order distance experiments.}
\label{fig_res_distance}
\end{figure*}

We implemented both OoO SWAG variants in C++: the baseline classic
B-tree augmented with aggregates and the finger B-tree aggregator
(\dsName{}). We present experiments with competitive min-arity values:
$2$, $4$ and $8$. Higher values for min-arity were never competitive
in our experiments. Our experiments run outside of any particular
streaming framework so we can focus on the aggregation algorithms
themselves. Our load generator produces synthetic data items with
random integers. The experiments perform rounds of \lstinline{evict},
\lstinline{insert}, and \lstinline{query} to maintain a sliding window
that accepts a new data item, evicts an old one, and produces a result
each round.

We present results with three aggregation operators $\otimes$ and
their corresponding monoids, each representing a different category of
computational cost. The operator \textsf{sum} performs an integer sum
over the window, and its computational cost is less than that of tree
traversals and manipulations. The operator \textsf{geomean} performs a
geometric mean over the window. For numerical stability, this requires
a floating point log on insertion and floating point additions during
data structure operations. It represents a middle ground in
computational cost. The most expensive operator, \textsf{bloom}, is a
Bloom filter~\cite{bloom_1970} where the partial aggregations maintain a bitset of size
$2^{14}$. It represents aggregation operators where the computational
cost of performing an aggregation easily dominates the cost of
maintaining the SWAG data structure.

We ran all experiments on a machine with an Intel Xeon E5-2697 at 2.7
GHz running Red Hat Enterprise Linux Server 7.5 with a 3.10.0 kernel.
We compiled all experiments with \lstinline{g++} 4.8.5 with
optimization level \lstinline{-O3}.

\subsection{Varying Distance}

We begin by investigating how \lstinline{insert}'s out-of-order
distance affects throughput. The distance varying experiments,
Figure~\ref{fig_res_distance}, maintain a window with a constant size
of $n=2^{22}=4,194,304$ data items. The $x$-axis is the out-of-order distance $d$
between the newest timestamp already in the window and the time\-stamp
created by our load generator. Our adversarial load generator
pre-populates the window with high timestamps and then spends the
measured portion of the experiment producing low timestamps. This
regime ensures that after the pre-population with high timestamps, the
out-of-order distance of each subsequent insertion is precisely~$d$.

This experiment confirms the prediction of the theory. The classic
B-tree's throughput is mostly unaffected by the change in distance,
but the finger B-tree's throughput starts out significantly higher and
smoothly degrades, following a $\log d$ trend. All variants see an
uptick in performance when $d=n$, that is, when the distance is the
size $n$ of the window. This is a degenerate special case. When $n=d$,
the lowest timestamp to evict is always in the left-most node in the
tree, so the tree behaves like a last-in first-out (LIFO) stack, and
inserting and evicting it requires no tree restructuring.

The min-arity that yields the best-performing B-tree varies with the aggregation
operator. For expensive operators, such as \textsf{bloom}, smaller min-arity
trees perform better. The reason is that as the min-arity grows, the number of
partial aggregations the algorithm needs to perform inside of a node also increases.
When the aggregation cost dominates all others, trees that require fewer total
aggregations will perform better. On the flip side, for cheap operators, such as
\textsf{sum}, trees that require fewer rebalance and repair
operations will perform better.

The step-like throughput curves for the finger B-trees is a function
of their min-arity: larger min-arity means longer sections where the
increased out-of-order distance still affects only a subtree with the
same height. When the throughput suddenly drops, the increase in $d$
meant an increase in the height of the affected subtree, causing more
rebalances and updates.

\subsection{Latency}

\begin{figure*}[!t]
\center
\includegraphics[width=0.32\textwidth]{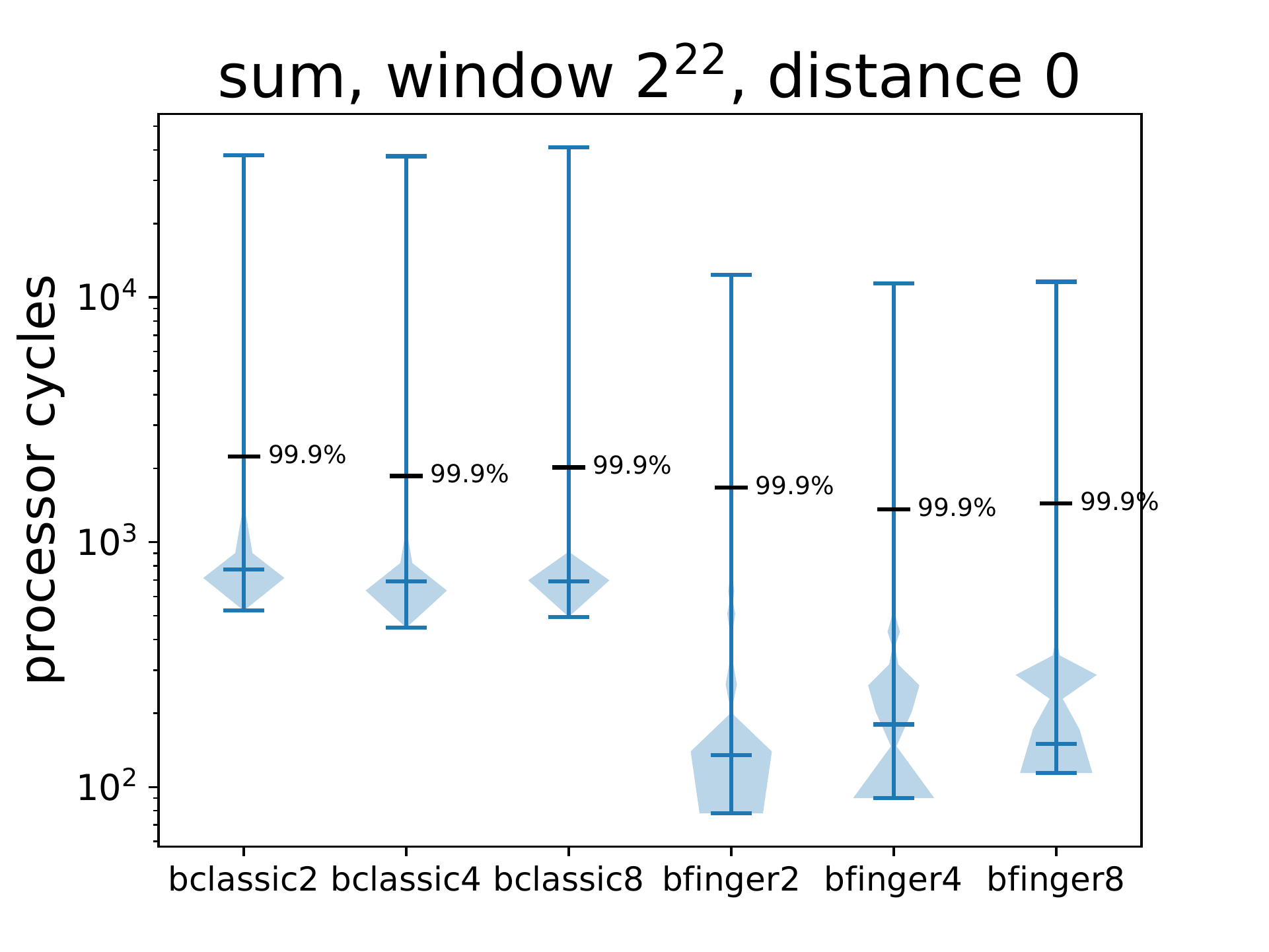}
\includegraphics[width=0.32\textwidth]{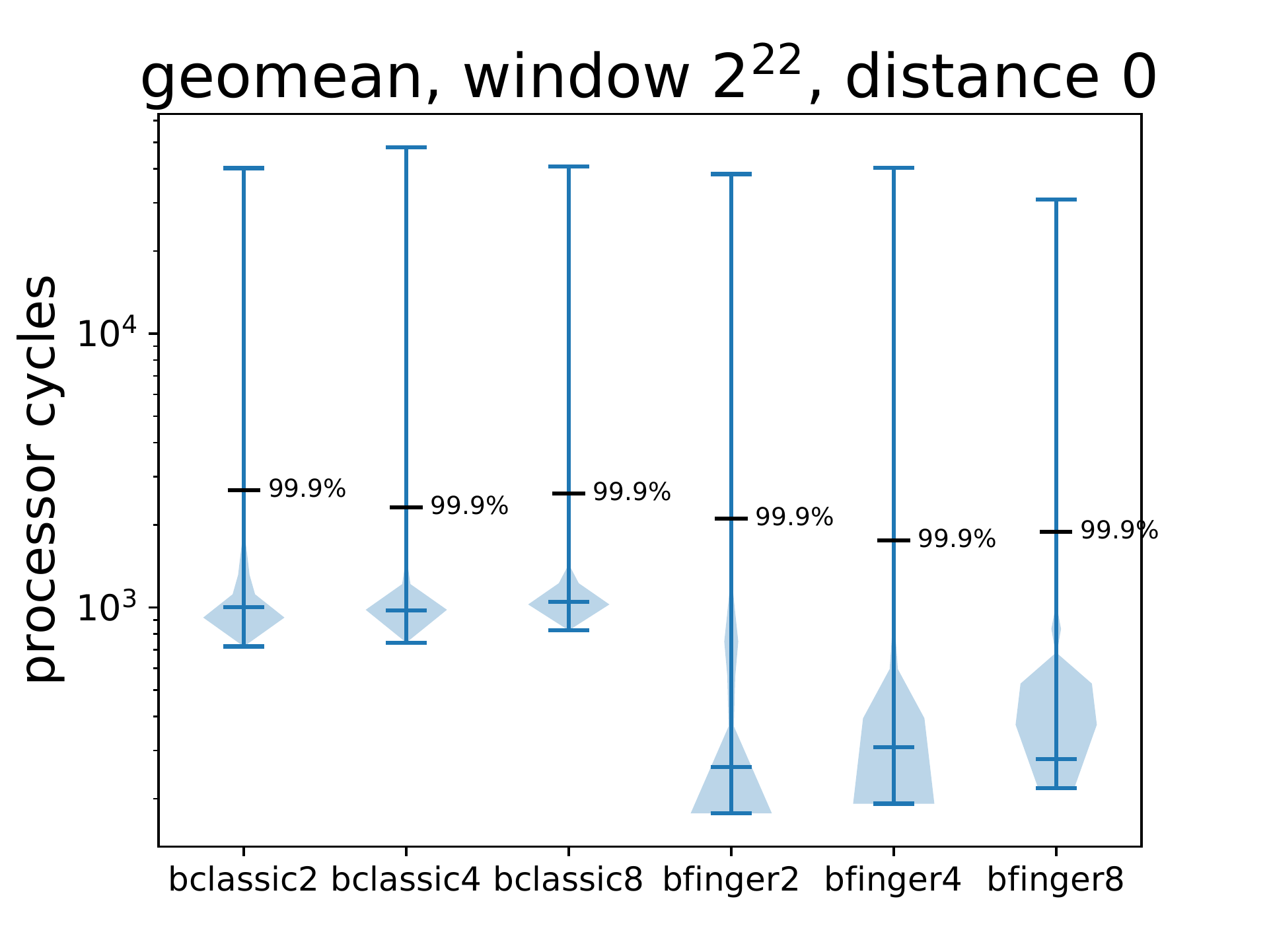}
\includegraphics[width=0.32\textwidth]{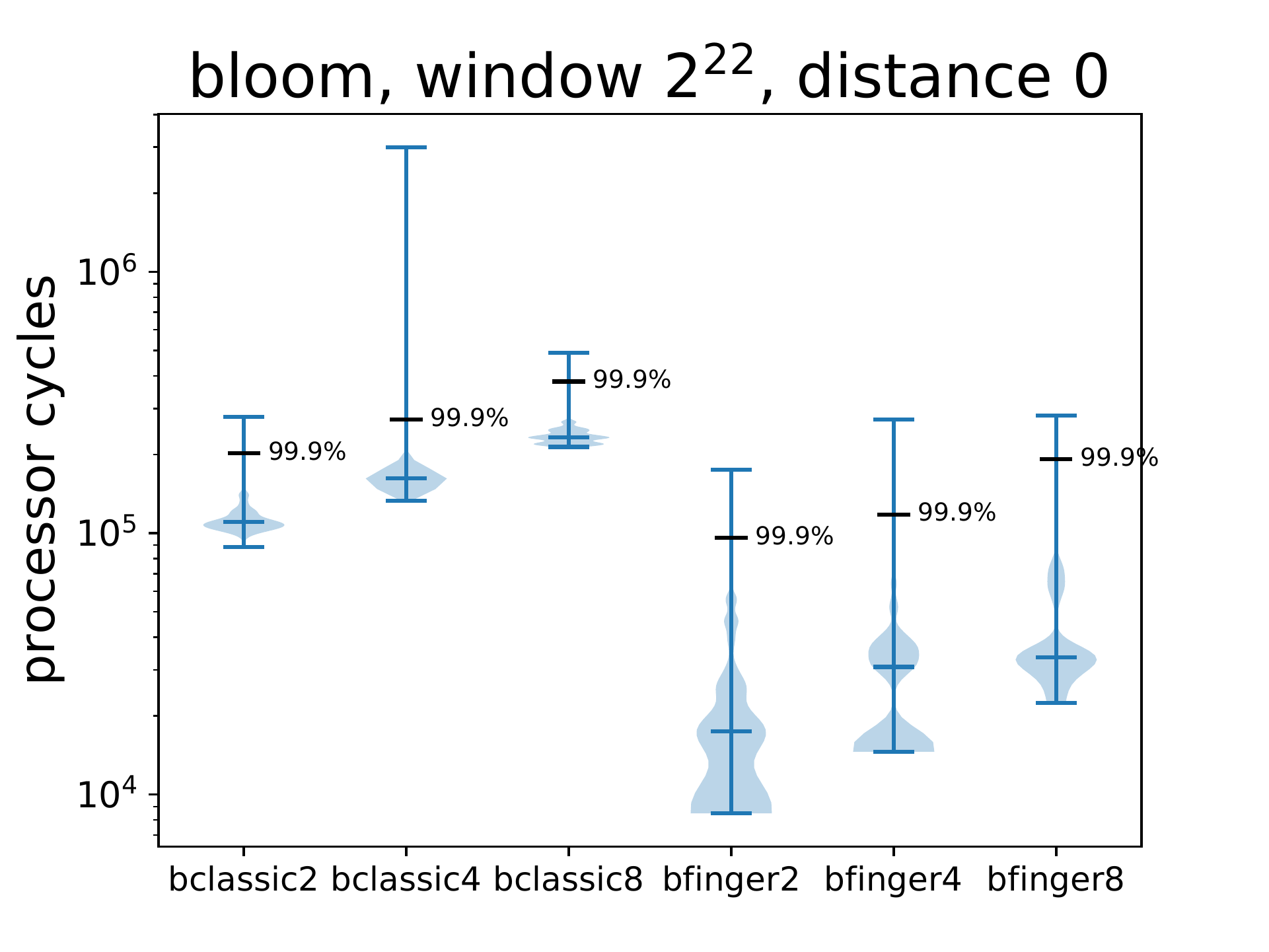}
\includegraphics[width=0.32\textwidth]{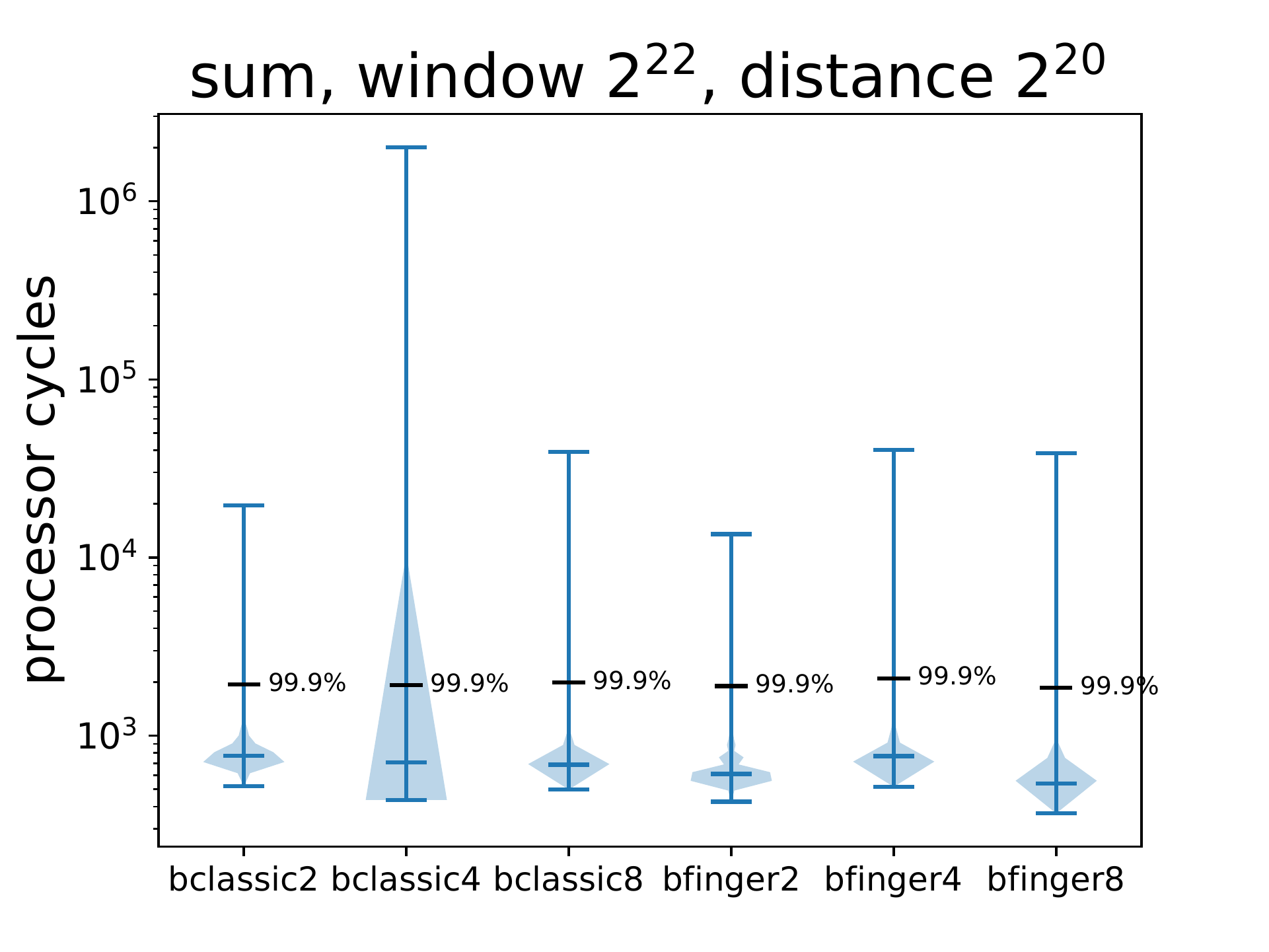}
\includegraphics[width=0.32\textwidth]{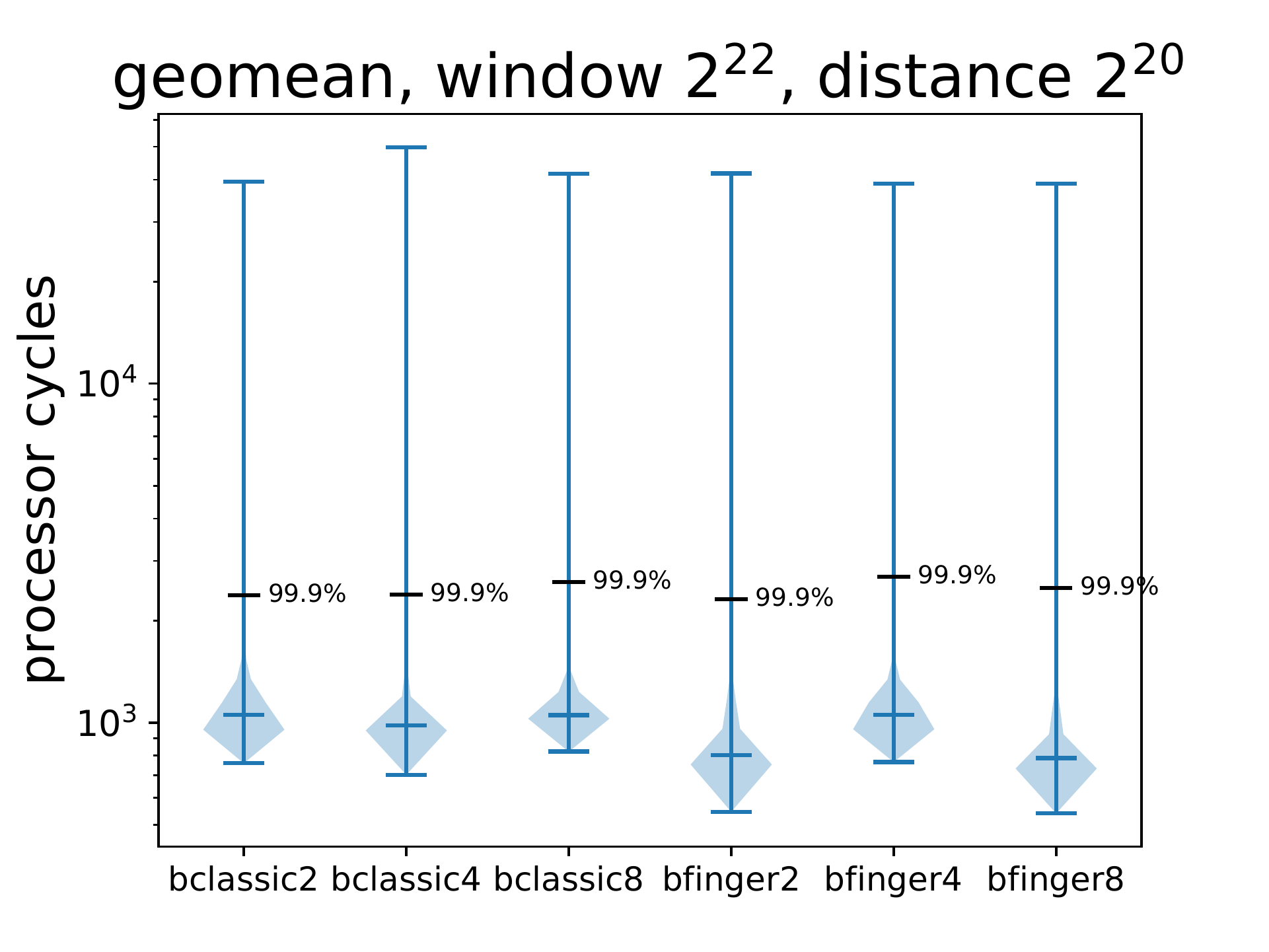}
\includegraphics[width=0.32\textwidth]{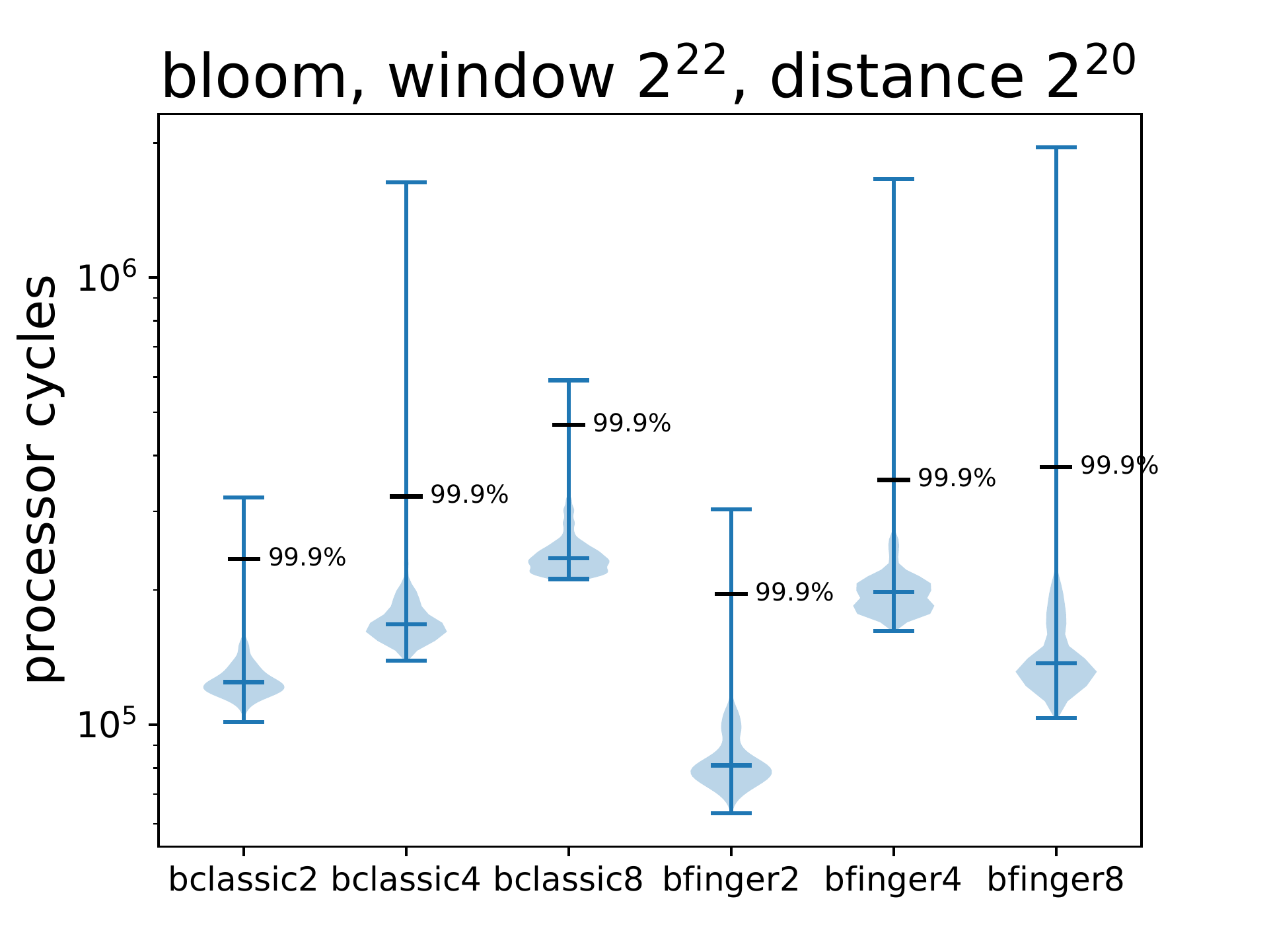}
\caption{Latency experiments.}
\label{fig_res_latency}
\end{figure*}

The worst-case latency for both classic and finger B-trees is $O(\log
n)$, but we expect that the finger variants should significantly
reduce average latency. The experiments in
Figure~\ref{fig_res_latency} confirm this expectation. All latency
experiments are with a fixed window of size $2^{22}$. The top set of
experiments use an out-of-order distance of $d=0$ and the bottom set
use an out-of-order distance of $d=2^{20}=1,048,576$. (We chose the latter
distance because it is among the worst-performing in the throughput
experiments.) The experimental setup is the same as for the throughput
experiments, and the latency is for an entire round of
\lstinline{evict}, \lstinline{insert}, and \lstinline{query}. The
$y$-axis is the number of processor cycles for a round, in log scale.
Since we used a 2.7 GHz machine, $10^3$ cycles take 370 nanoseconds
and $10^6$ cycles take 370 microseconds.
The blue bars represent the median latency, the shaded blue regions
represent the distribution of latencies, and the black bar is the
$99.9$th percentile. The range is the minimum and maximum latency.

When the out-of-order distance is~$0$ and the aggregation operator is
cheap or only moderately expensive, the worst-case latency in practice
for the classic and finger B-trees is similar. This is expected, as
the time is dominated by tree operations, and they are worst-case
$O(\log n)$. However, the minimum and median latencies are orders of
magnitude better for the finger B-trees.  This is also expected, since
in the case of $d=0$, the fingers enable amortized constant
updates. When the aggregation operator is expensive, the finger
B-trees have significantly lower latency, because they have to
repair fewer partial aggregates.

With an out-of-order distance of $d=2^{20}$ and cheap or moderately
expensive operators, the classic and finger B-trees have similar
latency. This is expected: as $d$ approaches $n$, the worst-case
latency for finger B-trees approaches $O(\log n)$. Again, with
expensive operators, the minimum, median, and $99.9$th percentile of
the finger B-tree with min-arity $2$ is orders of magnitude lower than
that of classic B-trees. There is, however, a curious effect clearly
present in the \textsf{bloom} experiments with finger B-trees, but
still observable in the others: min-arity $2$ has the lowest latency;
it gets significantly worse with min-arity $4$, then improves with
min-arity~$8$. Recall that the root is not subject to min-arity---in
other words, it may be slimmer. With $d=2^{20}$, depending on the
arity of the root, some aggregation repairs walk almost to the root
and then back down a spine while others walk to the root and no
further. The former case, which involves twice a spine, is generally more
expensive than the latter, which is usually a shorter path. The
frequency of the expensive case is a function of the window size, tree
arity, and out-of-order distance, and these factors do not interact
linearly.

\begin{figure*}[!t]
\center
\includegraphics[width=0.32\textwidth]{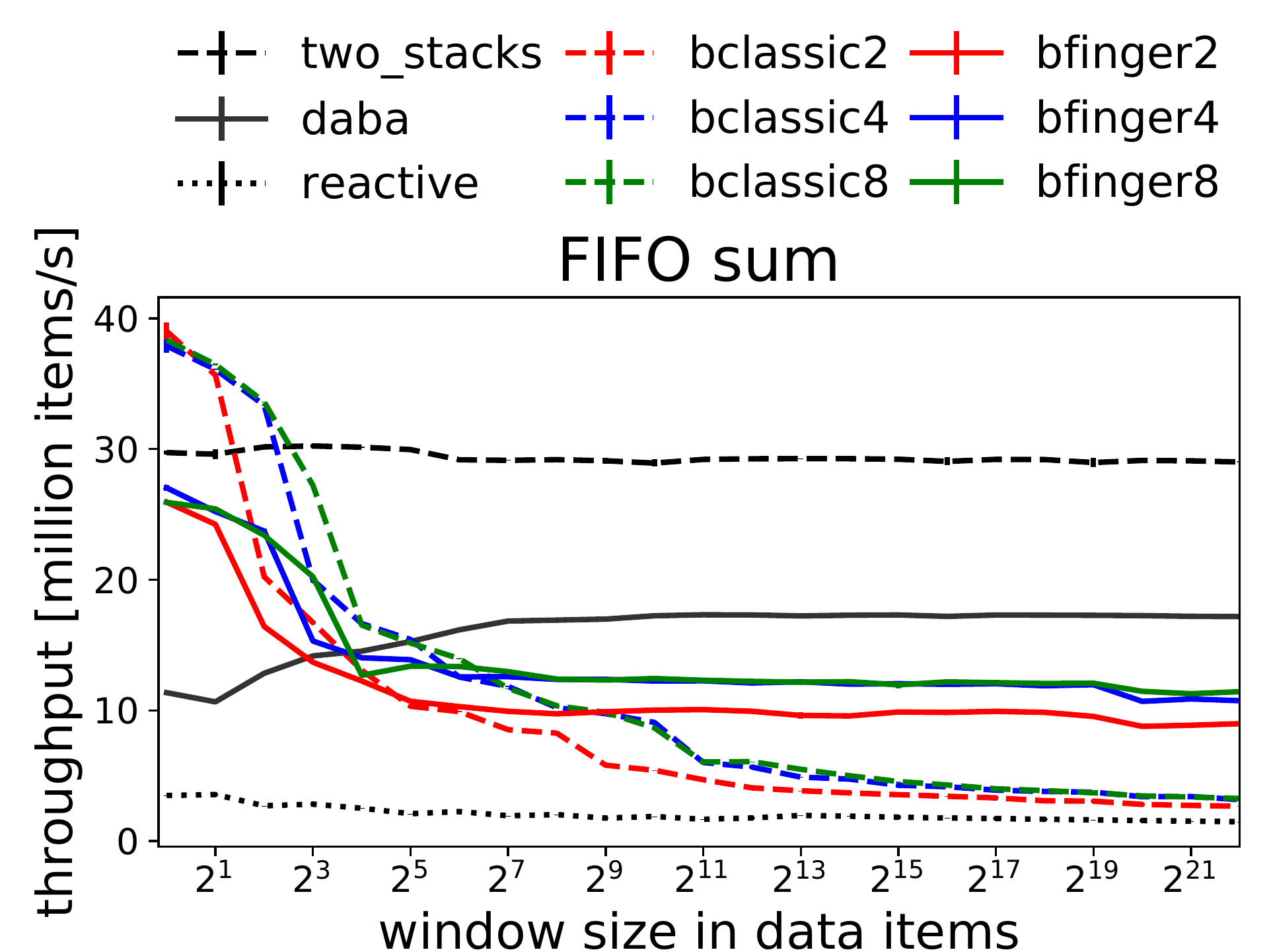}
\includegraphics[width=0.32\textwidth]{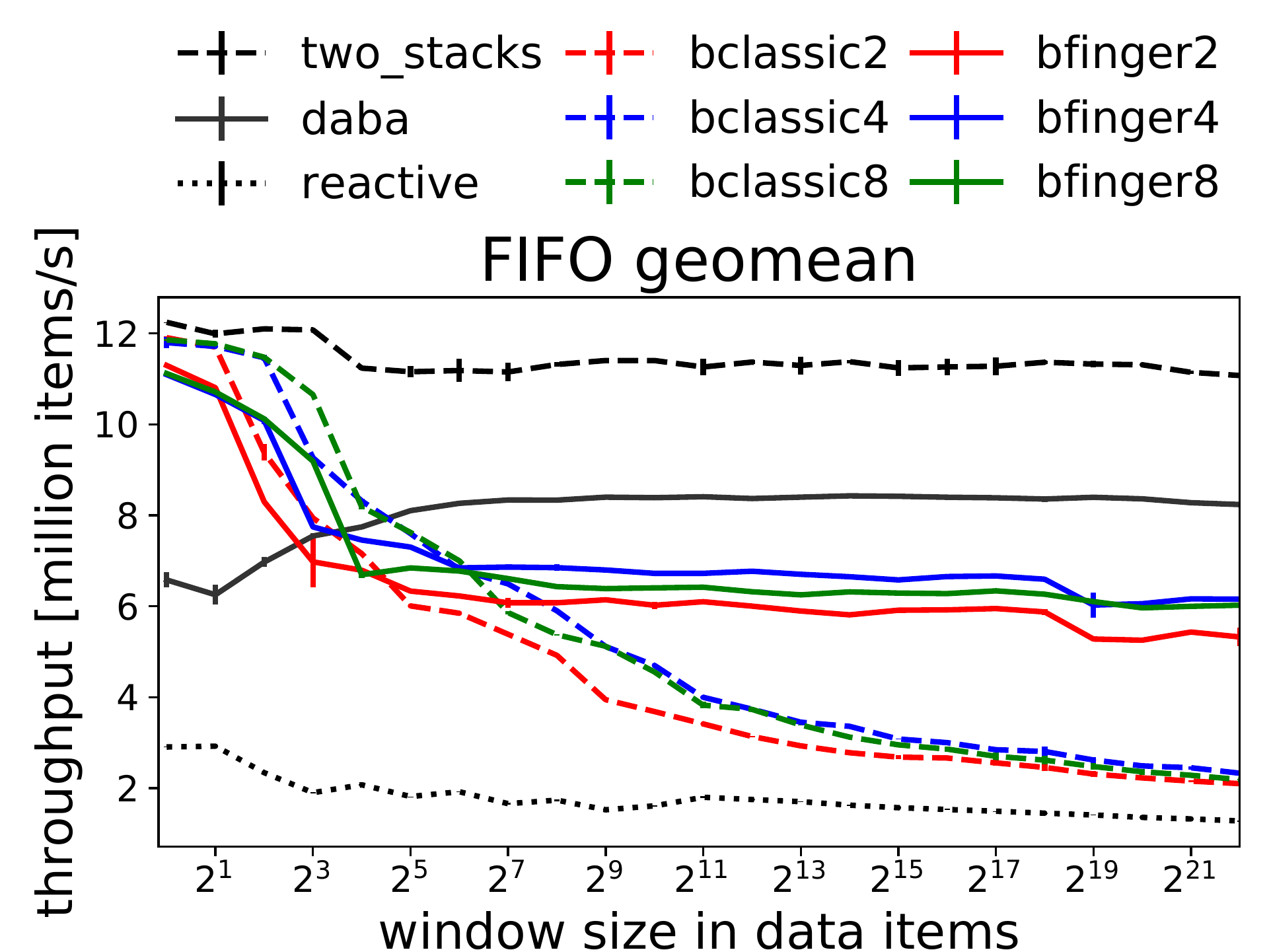}
\includegraphics[width=0.32\textwidth]{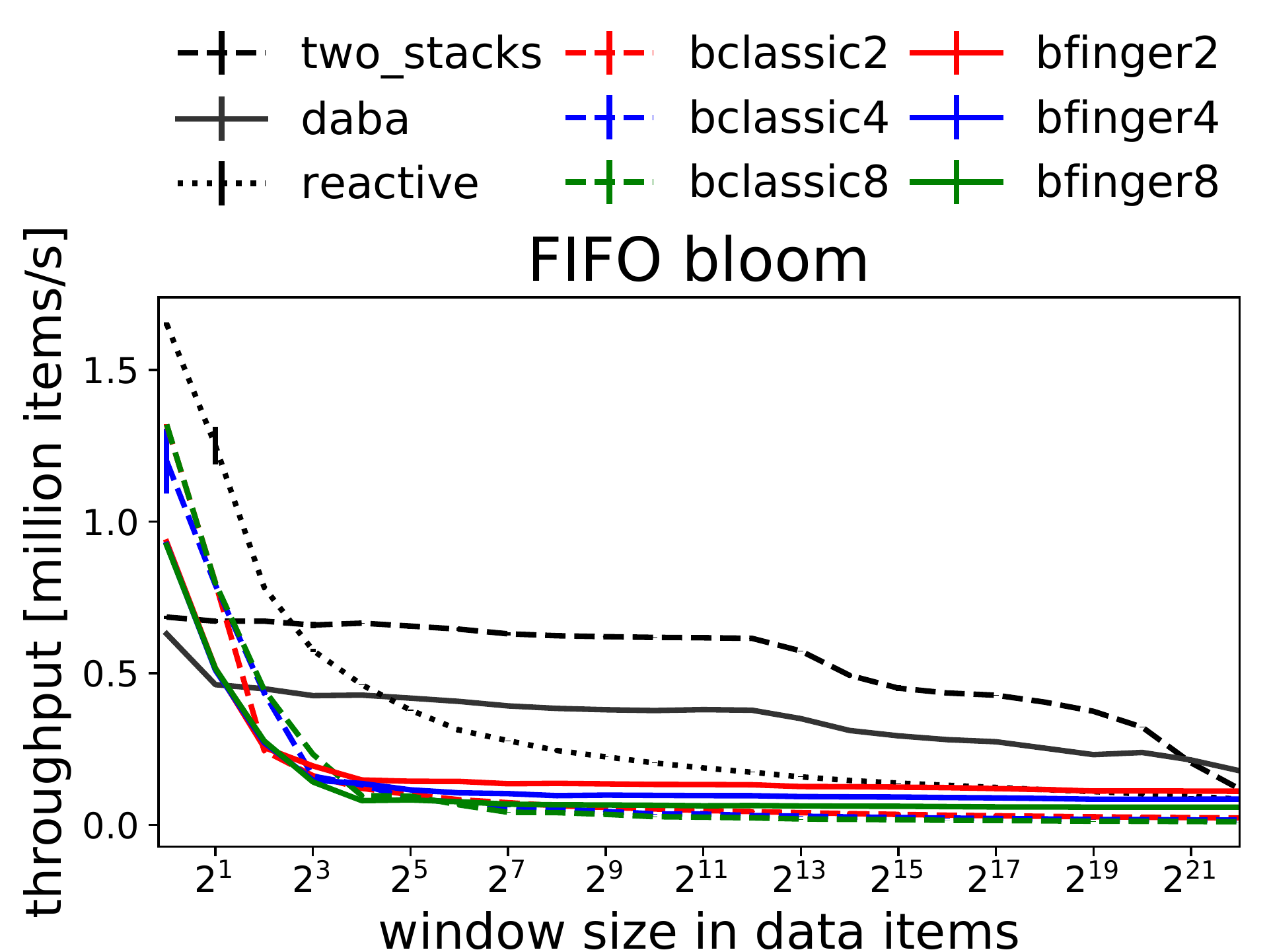}
\caption{FIFO experiments.}
\label{fig_res_fifo}
\end{figure*}

\subsection{FIFO}
A special case for \dsName{} is when $d=0$; with in-order data, our finger B-tree
aggregator (\dsName) enjoys amortized constant time performance. Figure~\ref{fig_res_fifo}
compares the B-tree-based SWAGs against the state-of-the art SWAGs optimized for
first-in, first-out, completely in-order data. Two-stacks only works on in-order
data and is amortized $O(1)$ with worst-case $O(n)$~\cite{adamax_2011}. The
De-Amortized Bankers Aggregator (DABA) also only works on in-order data and is
worst-case $O(1)$~\cite{tangwongsan_hirzel_schneider_2017}. The Reactive
Aggregator supports out-of-order evict but requires in-order insert and is
amortized $O(\log n)$ with worst-case~$O(n)$~\cite{tangwongsan_et_al_2015}. The
$x$-axis represents increasing window size~$n$.

Two-stacks and DABA perform as seen in prior work: for most window sizes,
two-stacks with amortized $O(1)$ time bound has the best throughput. DABA is
generally second best, as it does a little more work on each operation to
maintain worst-case constant performance.

The finger B-tree variants demonstrate constant performance as the window size
increases. The best finger B-tree variants stay within $30\%$ of DABA for
\textsf{sum} and \textsf{geomean}, but are about $60\%$ off of DABA with a more
expensive operator like \textsf{bloom}. In general, finger B-trees are able to
maintain constant performance with completely in-order data, but the extra work
of maintaining a tree means that SWAGs specialized for in-order data
consistently outperform them.

\begin{figure*}[!t]
\center
\includegraphics[width=0.32\textwidth]{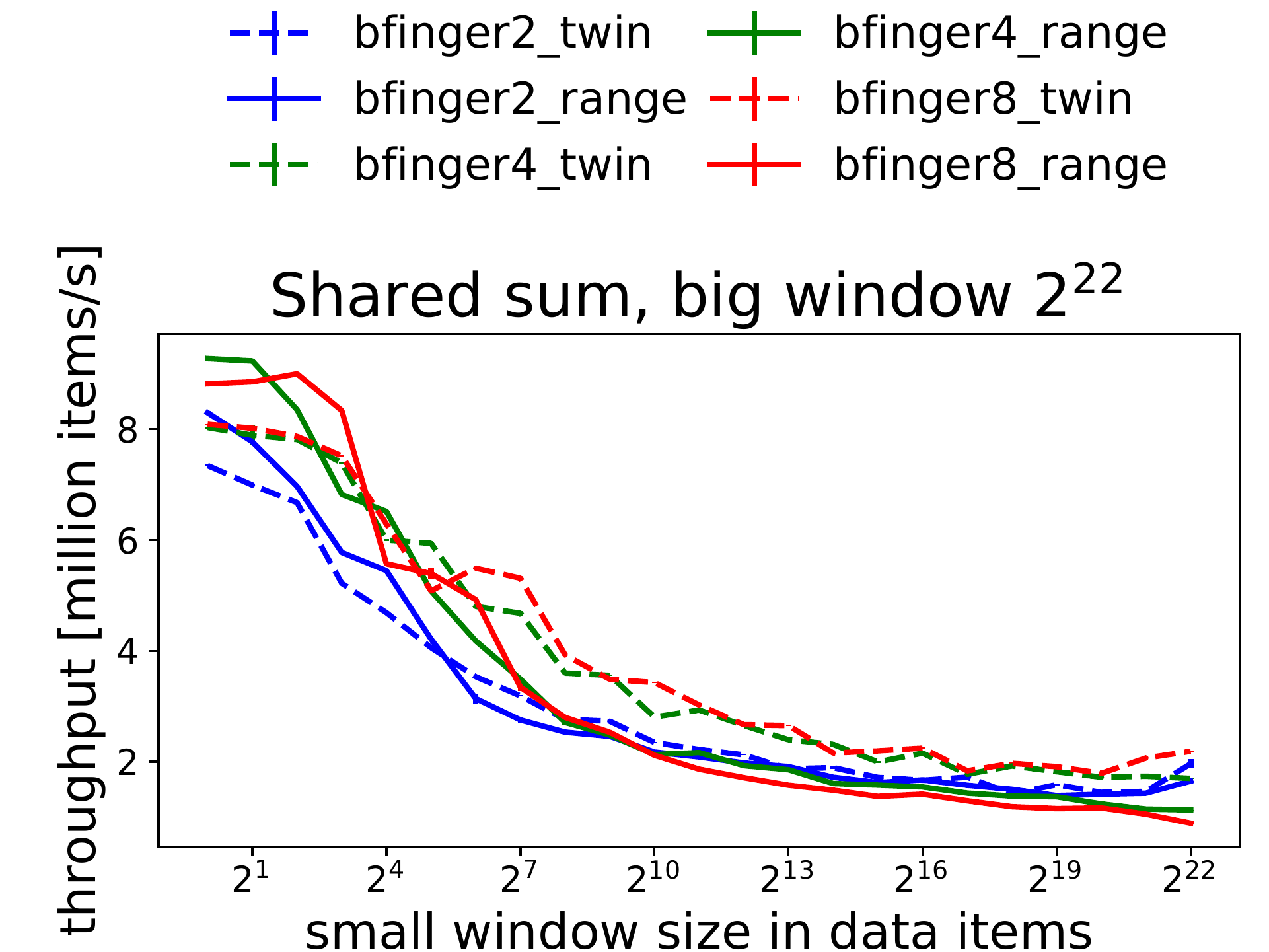}
\includegraphics[width=0.32\textwidth]{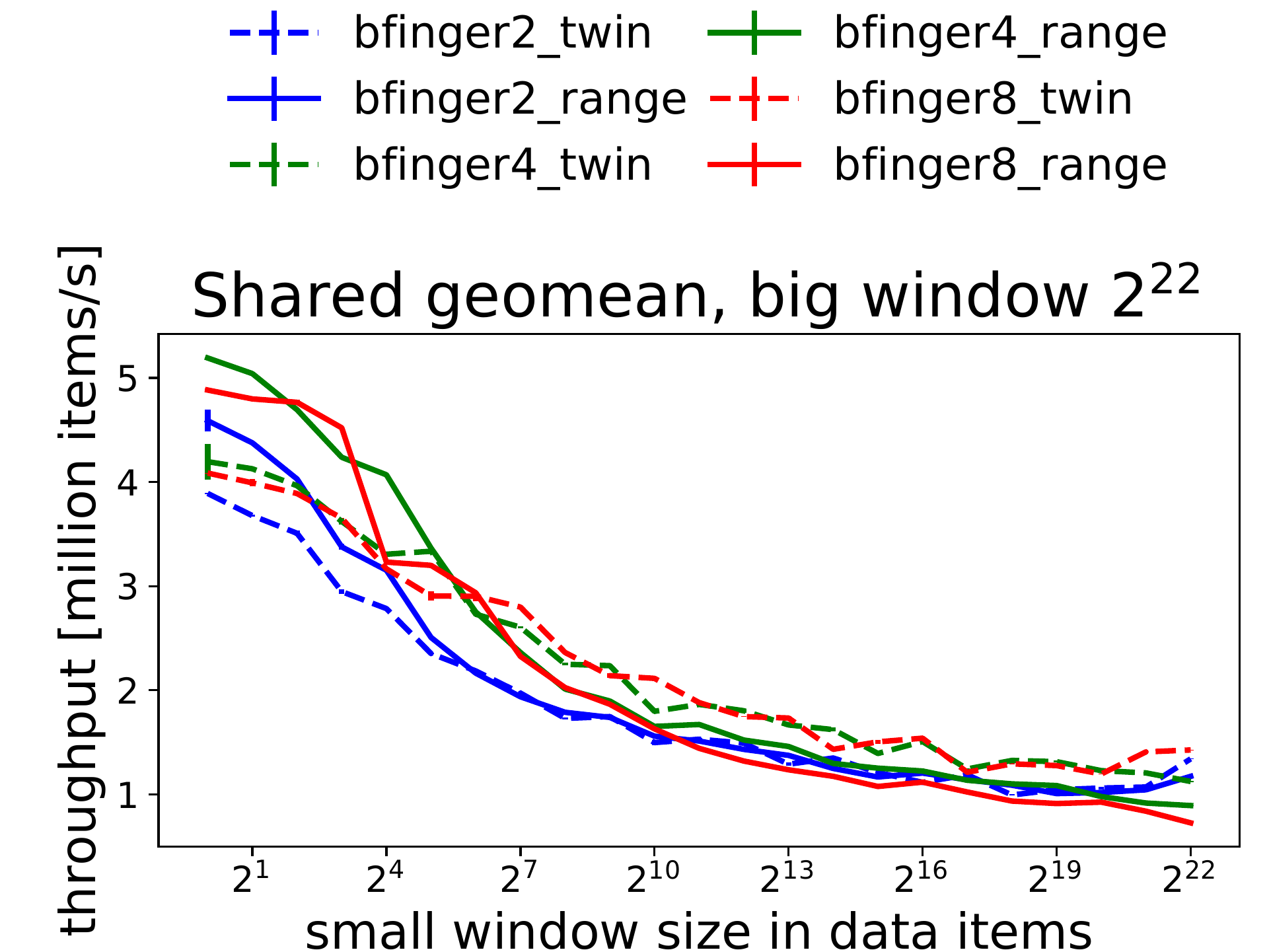}
\includegraphics[width=0.32\textwidth]{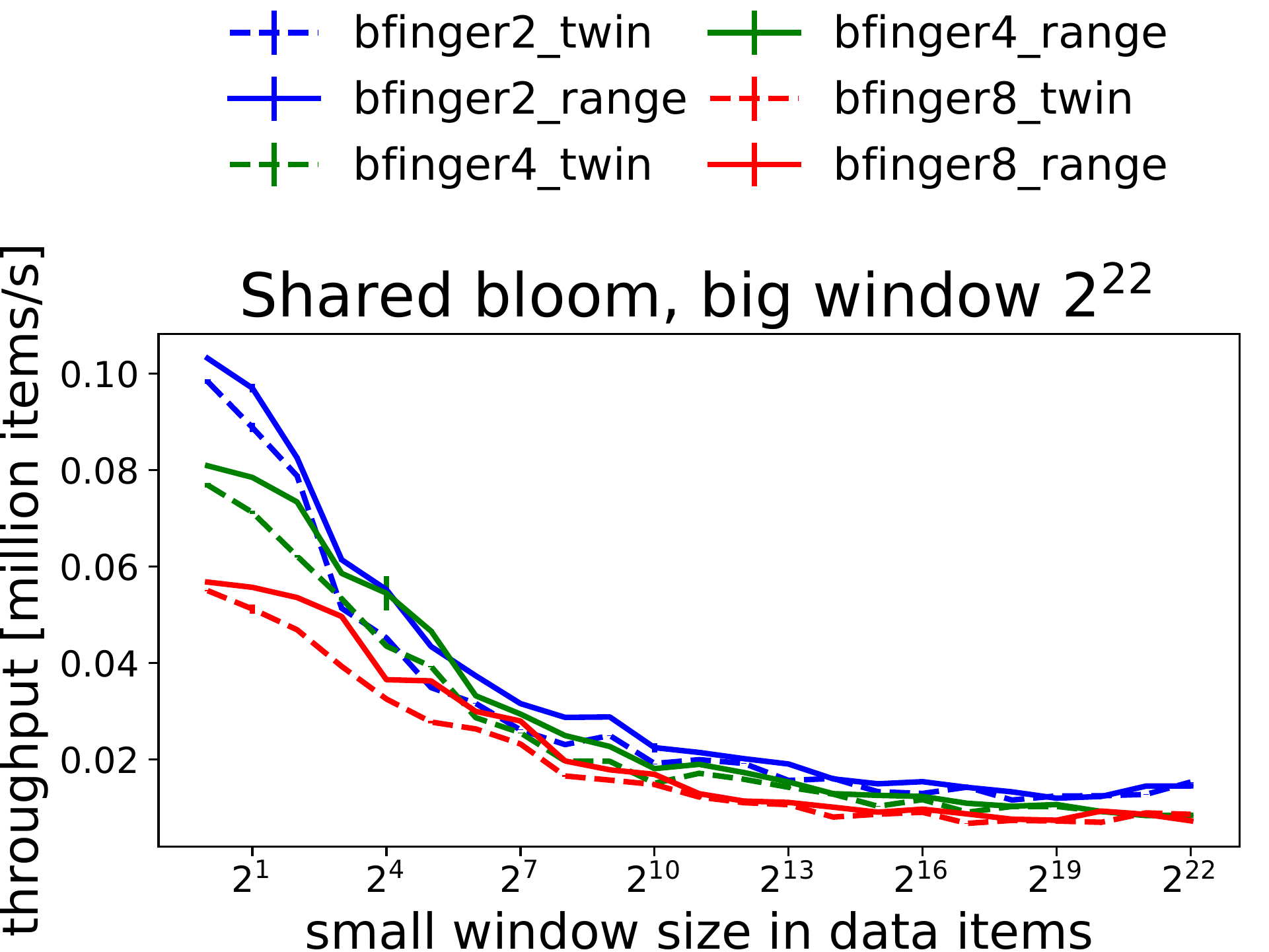}
\caption{Window sharing experiments. Out-of-order distance also varies as $n/2$ where $n$ is 
    the small window size.}
\label{fig_res_sharing}
\end{figure*}

The classic B-trees clearly demonstrate $O(\log n)$ behavior as the
window size increases. Reactive does demonstrate $O(\log n)$ behavior,
but it is only obvious with \textsf{bloom}. For \textsf{sum} and
\textsf{geomean}, the fixed costs dominate.  Reactive was designed to
avoid using pointer-based data structures under the premise that the
extra memory accesses would harm performance. To our surprise, this is
not true: on our hardware, the extra computation required to avoid
pointers ends up costing more. For \textsf{bloom},
Reactive outperforms all of the B-tree based SWAGs because it is
essentially a min-arity 1, max-arity 2 tree. As seen in other results,
for the most expensive aggregation operators, reducing the total number of
aggregation operations matters more to performance than data structure
updates.

\subsection{Window Sharing}

One of the benefits of finger B-trees is that they can support a
range-query interface while maintaining logarithmic performance for
queries over that range. A range-query interface enables window
sharing: the same window can be used for multiple queries over
different ranges. An obvious benefit from window sharing is reduced
space usage, but we also wanted to investigate if it could improve
runtime performance. As Figure~\ref{fig_res_sharing} shows, window
sharing did not consistently improve runtime performance.

The experiments maintain two queries: a big window fixed to size
$2^{22}$, and a small window whose size $n_\text{small}$ varies from
$1$ to $2^{22}$, shown on the $x$-axis. The workload consists of
out-of-order data items where the out-of-order distance $d$ is half of
the small window size, i.e., $d=n_\text{small}/2$. The \textsf{\_twin}
experiments maintain two separate trees, one for each window size. The
\textsf{\_range} experiments maintain a single tree, using a standard
query for the big window and a range query for the small window.

Our experiment performs out-of-order insert and in-order evict, so
insert costs $O(\log d)$ and evict costs $O(1)$. Hence, on average,
each round of the \textsf{\_range} experiment costs $O(\log d)$ for
insert, $O(1)$ for evict, and $O(1)+O(\log n_\text{small})$ for query
on the big window and the small window. On average, each round of the
\textsf{\_twin} experiment costs $2\cdot O(\log d)$ for insert,
$2\cdot O(1)$ for evict, and $2\cdot O(1)$ for query on the big and
small window. Since we chose $d=n_\text{small}/2$, this works out to a
total of $O(\log d)$ per round in both the \textsf{\_range} and the
\textsf{\_twin} experiments. There is no fundamental reason why window
sharing is slightly more expensive in practice. A more optimized code
path might make range queries slightly less expensive, but we would
still expect them to remain in the same ballpark.

By picking $d=n_\text{small}/2$, our experiments demonstrate the case
where window sharing is the most likely to outperform the twin
experiment. Since it did not outperform the twin experiment, we
conclude that window sharing is unlikely to have a consistent
performance benefit.  We could have increased the number of shared
windows to the point where maintaining multiple non-shared windows
performed worse because of the memory hierarchy, but that is the same
benefit as reduced space usage. We conclude that the primary benefits
of window sharing in this context are reduced space usage and the
ability to construct queries against arbitrarily sized windows
on the fly.


%% file: related.tex
\section{Related Work}
\label{sec:related}

This section describes work related to out-of-order sliding window
aggregation, sliding-window aggregation with window sharing, and
finger trees.

\myparagraph{Out-of-Order Stream Processing}
Processing out-of-order (OoO) streams is a popular research topic with a variety
of approaches. But there are surprisingly few incremental algorithms for OoO
stream processing. Truviso~\cite{krishnamurthy_et_al_2010} handles stream data
sources that are out-of-order with respect to each other but where input values
are in-order with respect to the stream they arrive on. The algorithm runs
separate stream queries on each source followed by consolidation. In contrast,
with \dsName{}, each individual stream input value can have its own
independent OoO behavior.
Chandramouli~et~al.~\cite{chandramouli_goldstein_maier_2010} describe how to
perform pattern matching on out-of-order streams but do not tackle sliding
window aggregation. Finally, the Reactive
Aggregator~\cite{tangwongsan_et_al_2015} performs incremental sliding-window
aggregation and can handle OoO evict in $O(\log n)$ time. In contrast,
\dsName{} can handle both OoO insert and OoO evict, and
takes sub-$O(\log n)$ time.

One approach to OoO streaming is \emph{buffering}: hold input stream values in a
buffer until it is safe to release them to the rest of the stream
query~\cite{srivastava_widom_2004}. Buffering has the advantage of not requiring
incremental operators in the query since the query only sees in-order data.
Unfortunately, buffering increases latency (since values endure non-zero delay)
and reduces quality (since bounded buffer sizes lead to outputs computed on
incomplete data). One can reduce the delay by optimistically performing
computation over transactional memory~\cite{brito_et_al_2008} and performing
commits in-order. Finally, one can tune the trade-off between quality and
latency by adaptively adjusting buffer sizes~\cite{ji_et_al_2015}. In contrast
to buffering approaches, \dsName{} can handle arbitrary
lateness without sacrificing quality nor significant latency.

Another approach to OoO streaming is \emph{retraction}: report outputs
quickly but revise them if they are affected by late-arriving inputs.
At any point, results are accurate with respect to stream input values
that have arrived so far. An early streaming system that embraced this approach
was Borealis~\cite{abadi_et_al_2005}, where stateful operators used
stored state for retraction. Spark Streaming also takes this approach:
it externalizes state from operators and handles stragglers like
failures, invalidating parts of the query~\cite{zaharia_et_al_2013}.
Pure retraction requires OoO algorithms such as OoO sliding window
aggregation, but the retraction literature does not show how to do
that efficiently, as the na\"ive approach of recomputing from scratch
would be inefficient for large windows. Our paper is complementary,
describing an efficient OoO sliding window aggregation algorithm that
could be used with systems like Borealis or Spark Streaming.

Using a \emph{low watermark} (lwm) is an approach to OoO streaming that combines
buffering with retraction. The lwm approach allows OoO values to flow through
the query but limits state requirements at individual operators by limiting the
OoO distance. CEDR proposed 8 timestamp-like fields to support a spectrum of
blocking, buffering, and retraction~\cite{barga_et_al_2007}. Li et
al.~\cite{li_et_al_2008} formalized the notion of a lwm based on the related
notion of punctuation~\cite{tucker_et_al_2003}. StreamInsight, which was
inspired by CEDR, offered a state-management interface to operator developers
that could be used for sliding-window aggregation. Subsequently,
MillWheel~\cite{akidau_et_al_2013}, Flink~\cite{carbone_et_al_2015},
and Beam~\cite{akidau_et_al_2015} also
adopted the lwm concept. The lwm provides some guarantees but leaves it to the
operator developer to handle OoO values. Our paper describes an efficient
algorithm for an OoO aggregation operator, which could be used with systems like
the ones listed above.

\myparagraph{Sliding Window Aggregation with Sharing}
All of the following papers focus on sharing over streams with the
same aggregation operator, e.g., monoid \mbox{$(S,\otimes,\onef{})$}.
The Scotty algorithm supports sliding-window aggregation over out-of-order
streams, while sharing windows with both different sizes and slice
granularities~\cite{traub_et_al_2018}.  For instance, Scotty might
share a window of size 60~minutes and granularity 3~minutes with a
session window whose gap timeout is set to 5~minutes. When a tuple
arrives out-of-order, older slices may need to be updated, fused, or
created.  Scotty
relies upon an aggregate store (e.g., based on a balanced tree) to
maintain slice aggregates. One caveat is that the aggregation
operator~$\otimes$ must be commutative; otherwise, one needs to keep
around the tuples from which a slice is pre-aggregated. Our \dsName{}
algorithm does not make any commutativity assumption. For commutative
operators, \dsName{} could serve as a more efficient
aggregate store for Scotty, thus combining the benefits of Scotty's
stream slicing with asymptotically faster final aggregation.

Other prior work on window sharing requires in-order
streams. The B-Int algorithm uses base intervals, which can be viewed
as a tree structure over ordered data, and supports sharing of windows
with different sizes~\cite{arasu_widom_2004}. Krishnamurthi et
al.~\cite{krishnamurthy_wu_franklin_2006} show how to share windows
that differ not just in size but also in granularity.
Cutty windows are a more efficient
approach to sharing windows with different sizes and
granularities~\cite{carbone_et_al_2016}, and their paper explains how to extend
the Reactive Aggregator~\cite{tangwongsan_et_al_2015} for sharing. The FlatFIT
algorithm performs sliding window aggregation in amortized constant time and
supports window sharing, addressing different granularities with the same
technique as Cutty windows~\cite{shein_chrysanthis_labrinidis_2017}. Finally,
the SlickDeque algorithm focuses on the special case where $x\otimes y$ always
returns one of either $x$ or $y$, and offers window sharing with $O(1)$ time
complexity assuming friendly input data
distributions~\cite{shein_chrysanthis_labrinidis_2018}.
In contrast to the above work, \dsName{} combines window sharing
with out-of-order processing.  It directly supports sliding window
aggregation over windows of different sizes.

\myparagraph{Finger Trees}
Our \dsName{} algorithm uses techniques from the literature on finger trees, combining and
extending them to work with sliding window aggregation. Guibas et
al.~\cite{guibas_et_al_1977} introduced finger trees in 1977. A \emph{finger}
can be viewed as a pointer to some position in a tree that makes tree operations
(usually search, insert, or evict) near that position less expensive. Guibas et
al.~used fingers on \mbox{B-trees}, but without aggregation. Huddleston and
Mehlhorn~\cite{huddleston_mehlhorn_1982} offer a proof that the amortized cost
of insertion or eviction at distance~$d$ from a finger is $O(\log d)$. Our proof
is inspired by Huddleston and Mehlhorn, but simplified and addressing a
different data organization: we support values to be stored at interior nodes,
whereas Huddleston and Mehlhorn's trees store values only in leaves. Kaplan and
Tarjan~\cite{kaplan_tarjan_1996} present a purely functional variant of finger
trees. The \emph{hands} data structure is an implementation of fingers that is
external to the tree, thus saving space, e.g., for parent
pointers~\cite{blelloch_maggs_woo_2003}. We did not adopt this techniques,
because in a B-tree, nodes are wider and thus, there are fewer nodes and
consequently fewer parent pointers in total. Finally, Hinze and
Paterson~\cite{hinze_paterson_2006} present purely functional finger trees with
amortized time complexity $O(1)$ at distance~1 from a finger. They describe
caching a monoid-based measure at tree nodes, but this cannot be directly used
for sliding-window aggregation. Our paper is the first to use finger trees for
fast out-of-order sliding window aggregation. The main novelty is to use and
maintain position-aware partial sums.


%% file: concl.tex
\section{Conclusion}
\label{sec:concl}

\dsName{} is a novel algorithm for sliding window aggregation
over out-of-order streams. The algorithm is based on finger B-trees
with position-aware partial aggregates. It works with any associative
aggregation operator, does not restrict the kinds of out-of-order
behavior, and also supports window sharing. This paper includes 
proofs of correctness and algorithmic complexity
bounds of our new algorithm. The proofs demonstrate that \dsName{} strictly
outperforms the prior state-of-the-art in theory and that it is as good as
the lower bound algorithmic complexity for this problem. In addition,
experimental results demonstrate that \dsName{} yields
excellent throughput and latency in practice.  Whereas in the past,
streaming applications that required out-of-order sliding window
aggregation had to make undesirable trade-offs to reach their
performance requirements, our new algorithm enables them to work
out-of-the-box for a broad range of circumstances.

%% file: appdx-lb.tex
\section{Running Time Lower Bound}
\label{sec:appdx-lb}
\newcommand*{\discrep}[0]{\delta}
\newcommand*{\permset}[0]{\mathcal{G}}  
This appendix proves Theorem~\ref{thm:swag-lb}, establishing a lower bound on any
OoO SWAG implementation. For a permutation $\pi$ on an ordered set $X$, denote
by $\pi_i$, $i = 1, \dots, |X|$, the $i$-th element of the permutation. Let
$\discrep_i(\pi)$ be the number of elements among $\pi_1, \pi_2, \dots,
\pi_{i-1}$ that are greater in value than $\pi_i$---that is, $\discrep_i(\pi) =
|\{j < i \mid \pi_j > \pi_i\}|$. This measure coincides with our notion of
out-of-order distance: if elements with timestamps $\pi_1, \pi_2, \dots$ are
inserted into OoO SWAG in that order, the $i$-th element has out-of-order distance
$\discrep_i(\pi)$.

For an ordered set $X$ and $d \geq 0$, let $\permset_d(X)$ denote the set of
permutations $\pi$ on $X$ such that $\max_i \discrep_i(\pi) \leq d$---i.e.,
every element is out of order by at most $d$. We begin the proof by bounding the
size of such a permutation set.

\begin{lemma}
  For an ordered set $X$ and $0 \leq d \leq |X|$,
  \[
    \left|\permset_d(X)\right|= d!(d+1)^{|X|-d}.
  \]
\end{lemma}
\begin{proof}
  The base case is $\left|\permset_0(\emptyset)\right| = 1$---the empty
  permutation. For non-empty $X$, let $x_0 = \min X$ be the smallest element in
  $X$. Then, every $\pi \in \permset_d(X)$ can be obtained by inserting $x_0$
  into one of the first $\min(|X|, d+1)$ indices of a suitable $\pi' \in
  \permset_d(X\setminus\{x_0\})$. In particular, each $\pi' \in
  \permset_d(X\setminus\{x_0\})$ gives rise to exactly $\min(|X|, d+1)$ unique
  permutations in $\permset_d(X)$. Hence, $\left|\permset_d(X)\right| =
  \left|\permset_d(X\setminus\{x_0\}) \right|\cdot \min(|X|, d+1)$. This expands
  to
  \[
    |\permset_d(X)| = \prod_{k=1}^{|X|} \min(k, d+1)
    = \Biggl(\prod_{k=1}^d k\Biggr)\Biggl(\prod_{k=d+1}^{|X|} d+1\Biggr), 
  \]
  which means $|\permset_d(X)|= d!(d+1)^{|X|-d}$,
    completing the proof.
\end{proof}

We will now prove Theorem~\ref{thm:swag-lb} by providing a reduction that sorts
any permutation $\pi \in \permset_d(\{1, 2, \dots, m\})$ using OoO SWAG.

\begin{proof}[Proof~of~Theorem~\ref{thm:swag-lb}]
  Fix $X = \{1, 2, \dots, m\}$. Let $A$ be a OoO SWAG implementation
  instantiated with the operator $x \otimes y = x$. When queried, this
  aggregation produces the first element in the sliding window.
  Now let $\pi$ be any permutation in $\permset_d(X)$. We will sort $\pi$ using
  $A$. First, insert $m$ elements $\tv{\pi_1}{\pi_1}, \tv{\pi_2}{\pi_2}, \dots,
  \tv{\pi_m}{\pi_m}$ into $A$. By construction, each insertion has out-of-order
  distance at most $d$. Then, query and evict $m$ times, reminiscent of heap sort.
  At this point, $\pi$ has been sorted using a total of $3m$ OoO SWAG
  operations.

  By a standard information-theoretic argument~(see, e.g.,
  \cite{cormen_leiserson_rivest_1990}), sorting a permutation in $\permset_d(X)$
  requires, in the worst case, $\Omega(\log |\permset_d(X)|)$ time. There
  are two cases to consider: If $d \leq \frac{m}2$, we have
  $|\permset_d(X)| \geq (1+d)^{m - d} \geq (1+d)^{m - m/2} = (1+d)^{m/2}$, so
  $\log |\permset_d(X)| \geq \Omega(m \log (1+d))$. Otherwise, we have $m \geq d
  > \frac{m}{2}$ and $|\permset_d(X)| \geq d! \geq (m/2)!$. Using Stirling's
  approximation, we know $\log |\permset_d(X)| = \Omega(m \log m)$, which is
  $\Omega(m \log (1+d))$ since $2m \geq 1+d$. In either case, $\log
  |\permset_d(X)| \geq \Omega(m \log (1+d))$.
\end{proof}


%% file: appdxproof.tex
\section{FiBA Correctness \& Complexity}\label{sec:appdxproof}

This appendix proves Theorem~\ref{trm_finger_correctness} (\dsName{}
correctness) and Theorem~\ref{trm_finger_complexity} (\dsName{}
algorithmic complexity).

\begin{proof}[Proof~of~Theorem~\ref{trm_finger_correctness}]
  There are two cases. If the root has no children (is a leaf), the inner
  aggregate stored at the root represents the aggregation of all the values
  inside the root node. Otherwise, by the aggregation invariants, we have the
  following observations: (1)~the aggregation at the right (left) finger is the
  aggregation of all values in the subtree that is the rightmost (leftmost) child of
  the root; and (2)~the aggregation at the root, represented by an inner
  aggregate, is the aggregation of all values in the tree excluding those
  covered by~(1).
  Therefore, \lstinline{query()}, which returns \lstinline{leftFinger.agg}
  $\otimes$ \lstinline{root.agg} $\otimes$ \lstinline{rightFinger.agg}, returns
  the aggregation of the values in the entire tree, in time order.
\end{proof}

\begin{proof}[Proof~of~Theorem~\ref{trm_finger_complexity}]
  The \lstinline{query()} operation performs at most \emph{two} $\otimes$
  operations; it clearly runs in $O(1)$ time.

  The search cost $T_\textrm{search}$ is bounded is follows. Let $y_0$ be the
  node at the finger where searching begins and recursively define $y_{i+1}$ as
  the parent of $y_i$. This forms a sequence of nodes on the spine on which
  searching takes place. Recall that
  \mbox{$\minarity\,=\,$\lstinline{MIN_ARITY}} is a constant. Because the
  subtree rooted at $y_i$ has $\Omega(\minarity^i)$ keys and the key we are
  searching is at distance $d$, we know the key belongs in the subtree rooted at
  some $y_{i^*}$, where $i^* = O(\log d)$. Thus, it takes $i^*$ steps to walk up
  the spine and at most another $i^*$ to locate the spot in the subtree as all
  leaves are at the same depth, bounding $T_\textrm{search}$ by
  $O(\log_\minarity d)$. The rebalance cost $T_\textrm{rebalance}$ is given by
  Lemma~\ref{lemma:restruct-cost} in the following section. Finally, following
  the aggregation invariants, a partial aggregation is affected only if it is
  along the search path or involved in rebalancing. Therefore, the number of
  affected nodes that requires repairs is bounded by $\Theta(T_\textrm{search} +
  T_\textrm{rebalance})$. Treating $\minarity$ as bounded by a constant,
  $T_\textrm{repair}$ is $O(T_\textrm{search} + T_\textrm{rebalance})$,
  concluding the proof.
\end{proof}


%% file: amortized_cost.tex
\section{Tree Rebalancing Cost}
To maintain the size invariants described earlier, all tree data structures
used in this paper require some restructuring after each update operation.
This section shows that restructuring only costs amortized $O(1)$ time
per operation.
More specifically, we prove the following lemma:
\begin{lemma}
  \label{lemma:restruct-cost}
  Let $\minarity \geq 2$. The amortized cost due to tree rebalancing in a B-tree
  with nodes of arity between \mbox{\lstinline{MIN_ARITY}$\,=\minarity$} and
  \mbox{\lstinline{MAX_ARITY}$\,=2\minarity$} (inclusive), starting with an empty
  tree initially, is $O(1)$ per OoO SWAG operation.
\end{lemma}

Important to the proof are the following observations:
\begin{itemize}
\item The only internal operations that can alter the tree structure are
  \lstinline{split}, \lstinline{merge}, \lstinline{move},
  \lstinline{heightIncrease}, and \lstinline{heightDecrease}. Each of these
  operations costs $O(1)$ time in the worst-case and only references nodes it
  has direct pointers to.

\item In a tree with minimum arity $\minarity$ and maximum arity $2\minarity$,
  during the intermediate steps, the arity of a node may be $\minarity - 1$ or
  $2\minarity + 1$. Hence, even in intermediate stages, each node always has arity
  between $\minarity-1$ and $2\minarity + 1$ (inclusive).
\item Finally, since $\minarity \geq 2$, we have $\minarity + 1 < \minarity + \minarity = 2\minarity$.
\end{itemize}

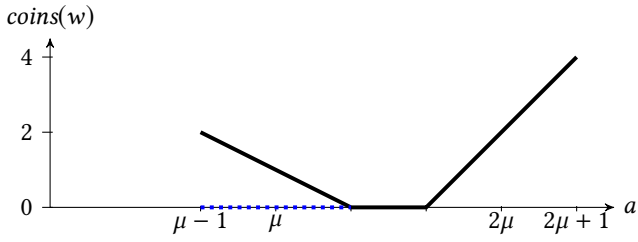
\begin{figure}[t]
  \input{figs/fig-coin-invariant}
  \caption{The black line shows the number of coins on a non-root node $w$ with
    arity~$a$. The number on the root node only differs in the blue dotted portion.}
  \label{fig:schematic-coin-invariant}
\end{figure}  

 \newcommand*{\hingeLoss}[1]{\bm{[}{#1}\bm{]}^+}
\begin{proof}[Proof~of~Lemma~\ref{lemma:restruct-cost}]
  This proof is a specialization of the rebalancing cost lemma
  in~\cite{huddleston_mehlhorn_1982}. We prove this lemma by showing that if
  each \lstinline{insert} and \lstinline{evict} is billed two coins, the
  following invariant on the amount of ``money'' can be maintained for every
  B-tree node. Let $w$ be a node with arity $a$.
  We aim for $w$ to maintain a reserve of $\textit{coins}(w)$ coins, where
  \[\textit{coins}(w)=\left\{\begin{array}{ll}
    4 & \mbox{if $a=2\minarity+1$}\\
    2 & \mbox{if $a=2\minarity$ or ($a=\minarity-1$ and $w$ is not the root)}\\
    1 & \mbox{if $a=\minarity$ and $w$ is not the root}\\
    0 & \mbox{if $a<2\minarity$ and ($a>\minarity$ or $w$ is the root)}
  \end{array}\right.\]
 as illustrated in Figure~\ref{fig:schematic-coin-invariant}.
  Now when \lstinline{insert} or \lstinline{evict} is called, the data structure
  locates a node in the tree where an entry is either added (due to
  \lstinline{insert}) or removed (due to \lstinline{evict}). In either case,
  $\textit{coins}(\cdot)$ of this node never changes by more than $2$,
  so $2$ coins suffice to cover the difference. However, this action
  may subsequently trigger
  a chain of \lstinline{split}s or \lstinline{merge}s. Below, we argue that the
  coin reserve on each node is sufficient to pay for such \lstinline{split}s and
  \lstinline{merge}s.
  
  When \lstinline{split} is called on a node $w$, it must be the case that $w$
  has arity $2\minarity + 1$. Therefore, $w$ itself has a reserve of $4$ coins.
  When $w$ is split, it is split into two nodes $\ell$ and $r$, with one entry
  promoted to \lstinline{$w$.parent}, the parent of $w$. Node $\ell$ will have arity $\minarity+1$
  and node $r$ will have arity $\minarity$. Because $\minarity < \minarity+1 < 2\minarity$,
  we have $\textit{coins}(\ell)=0$ and node $\ell$ needs no coin. But node $r$ has $\textit{coins}(r) =
  1$, so it will need $2$ coins. Moreover, now that the arity of \lstinline{$w$.parent} is
  incremented, node \lstinline{$w$.parent} may need up to 2 additional coins. Out of 4 coins $w$
  has, use $1$ to pay for the split, give $1$ to $r$, and give up to $2$ to
  \lstinline{$w$.parent}. Any excess is refunded.
  
  When \lstinline{merge} is called on a node $w$, it must be the case that $w$
  has arity $\minarity-1$ and the sibling it is about to merge with has
  arity~$\minarity$. Therefore, between these two nodes, we have $2 + 1 = 3$ coins in
  reserve. Once merged, the node has arity $\minarity + \minarity - 1 = 2\minarity-1$,
  so it needs $0$ coins. As a result of merging, the parent of $w$ loses one
  child, so it may potentially need $1$ coin. Out of $3$ coins in reserve, use
  $1$ to pay for the merge and give up to $1$ to \lstinline{$w$.parent}. Any excess is refunded.

  Finally, note that each of \lstinline{heightIncrease},
  \lstinline{heightDecrease}, and \lstinline{move} can take place at most once
  per a single OoO SWAG update. The internal operations
  \lstinline{heightIncrease} and \lstinline{heightDecrease} are easy to account
  for. For \lstinline{move}, when called on a node $w$, it must be the case that
  $w$ has arity $\minarity - 1$, and the sibling it is interacting with has arity $a'$,
  where $\minarity \leq a' \leq 2\minarity$. So, $w$ has $2$ coins. Once moved, $w$ has arity
  $\minarity$, so it needs only $1$ coin, leaving $1$ coin for the sibling to use. The
  sibling of $w$ will loose one arity, so it needs at most $1$ additional coin
  (either going from arity $\minarity + 1$ to $\minarity$, or $\minarity$ to
  $\minarity-1$). This concludes the proof.
\end{proof}


%% file: figs/fig-coin-invariant.tex
\tikzset{
  >=stealth',
  axis/.style={<->},
  curve line/.style={line width=1.5pt},
}
 \begin{tikzpicture}[xscale=1,yscale=0.5]
    \coordinate (y) at (0,4.5);
    \coordinate (x) at (7.5,0);
    \draw[<->] (y) node[above] {$\textit{coins}(w)$} -- (0,0) --  (x) node[right]
    {$a$};
    \foreach \x in {2,...,7}
      \draw (\x,1pt) -- (\x,-3pt)
        node[anchor=north] {};
    \foreach \y in {0,2,4}
      \draw (1pt,\y) -- (-3pt,\y) 
        node[anchor=east] {\y}; 

    \node at (2, -12pt)  {$\minarity-1$};
    \node at (3, -12pt)  {$\minarity$};
    \node at (6, -12pt)  {$2\minarity$};
    \node at (7, -12pt)  {$2\minarity+1$};

    \coordinate (left_end) at (2, 2);
    \coordinate (left_ground) at (4, 0);
    \coordinate (right_ground) at (5, 0);
    \coordinate (right_end) at (7, 4);
    \draw[black, curve line] (left_end) -- (left_ground) -- (right_ground) -- (right_end);
    \draw[blue, dotted, curve line] (2, 0) -- (left_ground);
  \end{tikzpicture}